\documentclass[a4paper,12pt]{article}

\usepackage[margin=2.5cm]{geometry}
\usepackage{amsmath,amssymb,amsthm}
\usepackage[bookmarks=false,linktocpage=true,colorlinks=true,linkcolor=blue,citecolor=blue,urlcolor=blue]{hyperref}

\numberwithin{equation}{section}

\newtheorem{definition}{Definition}
\newtheorem{theorem}{Theorem}

\newcommand{\p}{\partial}
\newcommand{\pb}{\bar\partial}
\newcommand{\zb}{{\bar z}}
\renewcommand{\AA}{{\mathbb A}}
\newcommand{\UU}{{\mathbb U}}
\newcommand{\Lc}{{\mathcal L}}

\newcommand{\QQ}{{\mathbb Q}}
\newcommand{\QQt}{{\widetilde {\mathbb Q}}}
\newcommand{\ZZ}{{\mathbb Z}}
\newcommand{\PP}{{\mathbb P}}
\newcommand{\EE}{{\mathbb E}}

\newcommand{\FF}{{\mathbb F}}
\newcommand{\FFt}{{\widetilde {\mathbb F}}}
\newcommand{\VV}{{\mathbb V}}
\newcommand{\lb}{[\![}
\newcommand{\rb}{]\!]}
\newcommand{\BB}{{\mathbb B}}
\newcommand{\GG}{{\mathbb G}}
\newcommand{\HH}{{\mathbb H}}

\newcommand{\N}{{\nabla}}
\renewcommand{\SS}{{\mathbb S}}
\newcommand{\Pc}{{\mathcal P}}

\renewcommand{\a}{{\alpha}}
\newcommand{\ah}{{\widehat\alpha}}
\renewcommand{\b}{{\beta}}
\newcommand{\bh}{{\widehat\beta}}
\renewcommand{\d}{{\delta}}

\newcommand{\ve}{{\varepsilon}}
\newcommand{\g}{{\gamma}}
\newcommand{\gh}{{\widehat \gamma}}
\newcommand{\G}{{\Gamma}}
\renewcommand{\i}{{\iota}}
\renewcommand{\l}{{\lambda}}
\newcommand{\lh}{{\widehat\lambda}}

\newcommand{\m}{{\mu}}
\newcommand{\mh}{{\widehat\mu}}
\newcommand{\n}{{\nu}}
\newcommand{\nh}{{\widehat\nu}}
\renewcommand{\o}{{\omega}}
\newcommand{\oh}{{\widehat\omega}}
\renewcommand{\O}{{\Omega}}
\newcommand{\Oh}{{\widehat\Omega}}
\newcommand{\s}{{\sigma}}
\newcommand{\sh}{{\widehat\sigma}}

\renewcommand{\t}{{\theta}}
\renewcommand{\th}{{\widehat\theta}}

\title{Pure Spinor String and Generalized Geometry}
\author{Dennis Zavaleta}
\date{
  \bigskip
  {\it\small Instituto de Fisica Te\'orica, Universidade Estadual Paulista} \\
  {\it\small Rua Dr. Bento Teobaldo Ferraz 271} \\
  {\it\small Bloco II - Barra Funda} \\
  {\it\small CEP:01140-070 - S\~ao Paulo, Brasil} \\
  [1 \baselineskip]
  {\it Email}: \href{mailto:dennis.zavaleta@unesp.br}{{\tt dennis.zavaleta@unesp.br}}
}

\begin{document}
\maketitle

  \begin{abstract}

  We consider the pure spinor sigma model in an arbitrary curved background.  The use of Hamiltonian formalism allows for a uniform description of the worldsheet fields where matter and ghosts enter the action on the same footing. This approach naturally leads to the language of generalized geometry. In fact, to handle the pure spinor case, we need an extension of generalized geometry. In this paper, we describe such an extension. We investigate the conditions for existence of nilpotent holomorphic symmetries. In the case of the pure spinor string in curved background, we translate the existing computations into this new language and recover previously known results.
  \end{abstract}

\section{Introduction}
The Pure Spinor formalism \cite{Berkovits2000} is a formulation of the superstring that allows for quantization with manifest super-Poincar\'e symmetry. In flat space, the action describes a free theory of matter fields $(Z^M) = (x^m, \t^\a, \th^\ah)$ and bosonic spinors (ghosts) $\l^\a$, $\lh^\ah$ subject to the pure spinor constraints $\l^\a \g^m_{\a\b} \l^\b = \lh^\ah \g^m_{\ah\bh} \lh^\bh = 0$. And also, there exist BRST charges $Q$ and ${\widehat Q}$ that encode the physical states of the theory in their cohomology.

It was shown in \cite{Berkovits2001} that given the Pure Spinor string in a general curved background, the conditions of nilpotency and holomorphicity of the BRST currents imply the Type II Supergravity equations of motion. One characteristic of the Pure Spinor sigma model is that the matter and ghost fields enter the action in different ways; this is,  while the matter fields have a second order kinetic term, the ghosts $\l^\a$ and $\lh^\ah$ are in first order form. Thus, the computations leading to the Type II SUGRA equations of motion did not show a clear geometrical interpretation.

In this paper we fill this gap by analyzing the Pure Spinor action in Hamiltonian form. The existence of this form of the action will initially depend on the invertibility of the Ramond-Ramond background field $P^{\a\ah}$ that appears in the original action but we will see later that this condition can be dropped and still describe a general Type II background. This version of the action will have as target space a graded supermanifold $M$ parametrized by matter and ghosts $(x^m, \t^\m, \th^\mh, \l^\a, \lh^\ah)$ where both type of fields are treated on the same footing.

In understanding the action of sigma models in Hamiltonian form, we need to make use of the concept of a generalized metric that first appeared in \cite{Hitchin2002} in the context of Generalized Complex Geometry \cite{Gualtieri2004}. The generalized metric was one of the elements used to extend Kahler geometry. Although, in the Physics literature, an equivalent structure had already appeared in the work of Gates, Hull and Rocek \cite{Gates1984} when studying $N=(2,2)$ supersymmetric sigma models (See \cite{Zabzine2006} for a review). In essence, for sigma models with target space $M$, a generalized metric gives a way of encoding both a Riemann metric $G$ and a two-form $B$ in a single tensor structure defined on $TM\oplus T^*M$ known as generalized tangent space. However, this ``generalized geometry" is limited in the sense that only allows to describe systems that posses a second order formulation leaving aside the ones with only a first order action or a mixture of both such as the pure spinor action. We will show that these sigma models with different formulations can be treated in a unified manner if we extend the original definition of a generalized metric to a one with a non-definite signature. Such an extension was given in \cite{Severa2018} but was only used to study sigma models with second order formulations having pseudo-Riemannian metrics.

When the symmetries of sigma models are studied, as well as their algebra of currents, there naturally appears the notion of Dorfman brackets $\lb\cdot,\cdot\rb$ \cite{Alekseev2004,Nekrasov2005}. It was shown in \cite{Severa2017} that given $S = \int (G+B)_{mn}\p x^m \pb x^n $, the conditions for $\d x^m = V^m(x)$ to be a symmetry (i.e. $\Lc_V G = 0$ and $\Lc_V B = dF$ for some one-form $F_m(x)$) become equivalent to imposing the preservation of the subbundles $graph(\pm G+B)\subset TM\oplus T^*M$ under the action of $\lb (V, F),\ \cdot\ \rb$. In the context of generalized geometry, these subbundles $graph(\pm G+B)$ are the $(\pm 1)$-eigenspaces $(TM\oplus T^*M)_\pm$ of a generalized metric. We will show that a similar situation occurs when we deal with the generalized metric with non-definite signature where these eigenspaces cannot be solved to be the graphs of some tensors anymore. Though still, the conditions of $(V,F)$ generating a symmetry will be equivalent to either the preservation of these eigenspaces or the fact that the Lie derivative of $(V,F)$ \cite{Kotov2010} on the generalized metric gives zero.

This paper is organized as follows. In section 2, we start with the sigma model of the bosonic string in curved space $M$ and study its action in Hamiltonian form. We see there is a matrix that characterizes the theory known as a {\it generalized metric} $\AA$ on $TM\oplus T^*M$. Then, we propose this form of the action as a starting point such as to include theories that do not admit a second order form. Lastly, for this generalization we obtain that the conditions to have holomorphic and nilpotent currents can be written roughly in terms of the Lie derivative of $\AA$ and the nilpotency of a section of $TM\oplus T^*M$ w.r.t $\lb\cdot,\cdot\rb$.

In section 3, we take the action of the Pure Spinor string in a general curved background and then transform it into its Hamiltonian form. Initially we assume that this is only valid for backgrounds with an invertible Ramond-Ramond field $P^{\a\ah}$ but later we show that this assumption can be dropped (since every term containing $P^{-1}_{\ah\a}$ gets cancelled). Moreover, we verify that the matrix $\AA_{PS}$ appearing in the pure spinor action (that encodes the vielbein and background fields $(\O, \Oh, C, {\widehat C}, P, S)$) satisfy the properties of a generalized metric. This section finishes by applying the nilpotency and holomorphicity conditions from the section prior to this one and rederiving the Type II SUGRA constraints. We finish discussing the implications that the ghost grading $(\l^\a,\lh^\ah)$ has on the generalized metric.

In section 4 we give our conclusions. We leave for the appendices the following subjects. A discussion on the relation between worldsheet conformal symmetry of the action and the fact that a generalized metric $\AA$ belongs to the orthosymplectic supergroup (when written as a super-matrix). Some theorems and proofs on about generalized metrics. Also, our conventions for super-geometry are given at the end.

\section{Holomorphic and nilpotent currents}
In this section we study the mathematical structures needed to write the conditions of holomorphicity and nilpotency of currents in terms of associated sections on $TM\oplus T^*M$.

\subsection{Sigma models in Hamiltonian form}
Let's start our discussion with the sigma model action of a string moving in a target space $M$ parametrized by coordinates $(x^m)$
  \begin{equation}
  S[x] = \frac12 \int \pb x^n \p x^m (G_{mn} + B_{mn})
  \end{equation}
It is known that a transformation $\d x^m = \ve V^m(x)$ is a symmetry of the action when $\Lc_V G = 0$ and $\Lc_V B = dF$ for some one-form $F_m(x)$. Such a symmetry provide us with a conserved current $(j_z,j_\zb)$, and the expression for each component can be computed using Noether's procedure
  \begin{align}
  j_z & =  \frac12 \p x^n (G_{nm} + B_{nm})V^m + \frac12 \p x^n F_n \label{jz} \\
  j_\zb & = \frac12 \pb x^n (G_{nm} - B_{nm})V^m - \frac12 \pb x^n F_n \label{jzb}
  \end{align}

We can also use the Hamiltonian form of the action which breaks explicit worldsheet conformal symmetry. The action in Hamiltonian form becomes
  \begin{equation}\label{hamil}
  S = \int \p_t x^m p_m - \frac12 (\p_\s x^m\ p_m)\left(
    \begin{array}{cc}
    (G -BG^{-1}B)_{mn} & -(BG^{-1})_m{}^n \\
    (G^{-1}B)^m{}_n & (G^{-1})^{mn}
    \end{array}
  \right) \left(
    \begin{array}{c}
    \p_\s x^n \\
    p_n
    \end{array}
  \right)
  \end{equation}
and the components of the Noether current (\ref{jz}) and (\ref{jzb}) are now
  \begin{align}
  j_z & = \frac12 (V\ F)\left[\left(
    \begin{array}{cc}
    I & 0 \\
    0 & I
    \end{array}
  \right) +  \left(
    \begin{array}{cc}
    -BG^{-1} & G-BG^{-1} B \\
    G^{-1} & G^{-1}B 
    \end{array}
  \right) \right] \left(
    \begin{array}{c}
    p \\
    \p_\s x
    \end{array}
  \right) \label{current-1} \\
  j_\zb & = \frac12  (V\ F) \left[ \left(
    \begin{array}{cc}
    I & 0 \\
    0 & I
    \end{array}
  \right) - \left(
    \begin{array}{cc}
    -BG^{-1} & G-BG^{-1} B \\
    G^{-1} & G^{-1}B
    \end{array}
  \right) \right] \left(
    \begin{array}{c}
    p \\
    \p_\s x
    \end{array}
  \right) \label{current-2}
  \end{align}
The matrices in the action and Noether current composed of the background fields $G,B,G^{-1}$ can be understood in terms of the so-called {\it generalized geometry}. Intuitively it can be understood as follows. Just in the same sense that Riemannian geometry deals with manifolds and metrics defined on their tangent spaces, {\it generalized} geometry will care about a manifold $M$ with a {\it generalized} metric defined on its {\it generalized} tangent space $TM\oplus T^*M$.

The matrix appearing in the action (\ref{hamil}) is what we will call the matrix representation of the generalized metric
  \begin{equation}\label{gen-met-2nd}
  A := \left(
    \begin{array}{cc}
    (G -BG^{-1}B)_{mn} & -(BG^{-1})_m{}^n \\
    (G^{-1}B)^m{}_n & (G^{-1})^{mn}
    \end{array}
  \right)
  \end{equation}
while the actual generalized metric is (in an abstract sense) a symmetric bilinear form on $TM\oplus T^*M$.
This matrix has the following defining%
\footnote{
  These are defining in the sense that any bilinear form $A$ on $TM\oplus T^*M$ with those properties will necessarily be of the form (\ref{gen-met-2nd}). This is, it will be associated to tensors $G_{mn}$ and $B_{mn}$.
}\label{footnote}
properties: it is symmetric and positive-definite, and it satisfies the $O(d,d)$ condition
  \begin{equation}\label{odd-condition}
  A^t \left(
    \begin{array}{cc}
    0 & 1 \\
    1 & 0
    \end{array}
  \right) A = \left(
    \begin{array}{cc}
    0 & 1 \\
    1 & 0
    \end{array}
  \right)
  \end{equation}
where $d$ is the dimension of the target space $M$. Furthermore, to the generalized metric we can associate an endomorphism $U_A$ on $TM\oplus T^*M$ with matrix representation given by
  \begin{equation}\label{invo-operator}
  U_A := A \left(
    \begin{array}{cc}
    0 & 1 \\
    1 & 0
    \end{array}
  \right) = \left(
    \begin{array}{cc}
    -BG^{-1} & G-BG^{-1} B \\
    G^{-1} & G^{-1}B
    \end{array}
  \right)
  \end{equation}
This matrix is the one appearing in both components of the Noether current (\ref{current-1}-\ref{current-2}). The operator $U_A$ has the property of being an involution i.e.
  \begin{equation}
  U_A^2 = \left(
    \begin{array}{cc}
    1 & 0 \\
    0 & 1
    \end{array}
  \right)
  \end{equation}
which is a direct consequence of the $O(d,d)$ condition (\ref{odd-condition}) that $A$ satisfies.

So far what we have done is to take a sigma model in second order formalism and transform its action into its Hamiltonian form. This process can be reversed by solving the equations of motion for $P_m$. However, we are interested in theories that do not admit a second order formulation (i.e. the ones that are purely first order or a mixture of both which the Pure Spinor string is an example of) i.e. theories where the generalized metric $A$ is not of the form (\ref{gen-met-2nd}) anymore. To study these types of theories we need to modify the definition of generalized metric.

The original definition of a generalized metric \cite{Hitchin2002} assumes that $A$ is of the form (\ref{gen-met-2nd}). This assumption can be understood in the following way. When we consider the associated operator to the generalized metric $A$
  \begin{equation}\label{associated}
  U_A := A \left(
    \begin{array}{cc}
    0 & 1 \\
    1 & 0
    \end{array}
  \right) = \left(
    \begin{array}{cc}
    A_m{}^n & A_{mn} \\
    A^{mn} & A^m{}_n
    \end{array}
  \right)
  \end{equation}
we have that is an involution i.e. $(U_A)^2 = 1$, and so we have $(\pm 1)$-eigenspaces which are defined as
  \begin{equation}\label{bundles}
  (TM\oplus T^*M)_\pm := \left\{ (V^m\ F_m)\ |\ (V^m\ F_m) \left(
    \begin{array}{cc}
    A_m{}^n & A_{mn} \\
    A^{mn} & A^m{}_n 
    \end{array}
  \right) = \pm (V^n\ F_n) \right\}
  \end{equation}
providing a way to express the bundle $(TM\oplus T^*M)$ as a direct sum $(TM\oplus T^*M) = (TM\oplus T^*M)_+ \oplus (TM\oplus T^*M)_-$. It turns out that these eigenspaces can always be solved as the graphs of some operators $(\pm G+B)$ i.e. 
  \begin{equation}
  (TM\oplus T^*M)_\pm = graph(\pm G+B):=\{ (\xi^m\ \a_m)\ |\ \a_m = \xi^n (\pm G_{nm} + B_{nm}) \}
  \end{equation}
for $G$ a metric tensor and $B$ a two-form. This ultimately implies that the action possesses a second order formulation (See Appendix \ref{gen-met} for a proof). We will see that in the case of the Pure Spinor sigma model, due to the target space being a graded supermanifold, it becomes impossible to represent these eigenspaces as graphs of some operators $\pm G +B$ without some entries in $G$ and $B$ having negative powers of $\l^\a$ and $\lh^\ah$. The easiest way to see this is to consider the case of flat space which is analyzed in section \ref{flat}.

In particular we must relax the condition of positivity of $A$ to admit an arbitrary signature. Thus, we define%
  \footnote{
  There are three different but equivalent definitions of a generalized metric. These equivalencies are also extended for the generalized metric with indefinite signature. We refer the reader to Appendix \ref{gen-met}.
  } %
the generalized metric with signature $(d+a,d-a)$ \cite{Severa2018} as a bilinear form in $TM\oplus T^*M$ such that its matrix representation $A$ satisfies
  \begin{enumerate}
  \item being symmetric
  \item having signature $(d+a, d-a)$
  \item the $O(d,d)$ property
  \end{enumerate}
Notice that for the value $a = d$ the signature becomes $(2d, 0)$ which is the same as positivity, and we are back to the original definition. However, unlike in the positive definite case, the eigenspaces $(TM\oplus T^*M)_\pm$ are not graphs of some tensors in general. For this to happen, the following two conditions must be satisfied. First of all, the dimension of $(TM\oplus T^*M)_\pm$ should be equal to $d = dim(TM) = dim(T^*M)$; and secondly, their intersections with the tangent $TM$ and contangent bundle $T^*M$ should be trivial. Equivalently,
  \begin{align}
  (TM\oplus T^*M)_\pm = graph(\pm G_{\bullet\bullet} + B_{\bullet\bullet}) \quad & \Leftrightarrow\quad A^{mn}\ \text{invertible}\\
  (TM\oplus T^*M)_\pm = graph(\pm R^{\bullet\bullet} + S^{\bullet\bullet}) \quad & \Leftrightarrow \quad A_{mn}\ \text{invertible}
  \end{align}
where $G$ and $R$ are symmetric while $B$ and $S$ anti-symmetric. From the sigma model point of view, we are only interested in the first case since this is the one that determines if the action has a second order formulation.

Finally, we claim that, for the bosonic case, the following setting is the correct generalization to study in a unified manner sigma models that mix second and first order formulations. We start with an action already in Hamiltonian form characterized by a generalized metric $A$ of indefinite signature
  \begin{equation}
  S[x,p] = \int \p_t x^m P_m - \frac12 (\p_\s x^m\ P_m)\left(
    \begin{array}{cc}
    A_{mn} & A_m{}^n \\
    A^m{}_n & A^{mn}
    \end{array}
  \right) \left(
    \begin{array}{c}
    \p_\s x^n \\
    P_n
    \end{array}
  \right)
  \end{equation}
Included in the definition of generalized metric is the $O(d,d)$ property (\ref{odd-condition}) that implies that this action is $2d$ conformally invariant. Furthermore, depending on the rank of the matrix $A^{mn}$, the action can be rewritten as a manifestly conformal invariant theory that mixes a first and second order formulation
  \begin{equation}
  S \sim \int \Pi_z^k \Pi_\zb^\ell (\eta_{k\ell} + b_{k\ell} ) + \frac12 \Pi_z^\ell p_{\zb \ell} + \frac12 \Pi_\zb^\ell p_{z \ell}\ ,\quad k,\ell=1,\ldots,rk[A^{mn}]
  \end{equation}
where $A^{mn} = e^m_k \eta^{k\ell} e^n_\ell$ and $\Pi_{z,\zb}^\ell = \p_{z,\zb} x^m e_m^\ell$, $(p_{z,\zb})_\ell=e^m_\ell(P_{t m}\pm P_{\s m})$. See Appendix \ref{app-ortho} for details on the relation of the $O(d,d)$ condition and conformal invariance.

It also happens that the group $O(d,d)$ acts%
  \footnote{
  This action of $O(d,d)$ on the space of generalized metrics is not in general a symmetry of the worldsheet theory.
  }
on the space of generalized metrics. Given $\mathcal O\in O(d,d)$, the generalized metric $A$ transforms as
  \begin{equation}
  A \quad \longrightarrow \quad \mathcal O^t A \mathcal O
  \end{equation}
In particular, it can be shown that the space of positive-definite generalized metrics coincides with the orbit of the unit matrix under this action i.e. $A = \mathcal O^t \mathcal O$. And, for cases where $A$ has indefinite signature, we need to consider other orbits.

To deal with the symmetries of this action in Hamiltonian form, we define a current $(j_z,j_\zb)$ associated to a section $(V^m\ F_m)$ of $TM\oplus T^*M$ following the form of (\ref{current-1}-\ref{current-2})
  \begin{equation}
    \begin{aligned}
    j_z & := (V\ F) \Pc_+ \left(
      \begin{array}{c}
      P \\
      \p_\s x
      \end{array}
    \right) \\
    j_\zb & := (V\ F) \Pc_- \left(
      \begin{array}{c}
      P \\
      \p_\s x
      \end{array}
    \right)
    \end{aligned}
  \end{equation}
where $\Pc_\pm$ are the projectors%
\footnote{
  This projectors are acting on the section $(V\ F)$ from the right. We chose this convention to emulate the one that we will use for the super-case in the next section.
}
onto the ($\pm 1$)-eigenspaces (\ref{bundles}) of $U_A$
  \begin{equation}\label{projectors}
  \Pc_\pm := \frac12( I \pm U_A) = \frac12 \left[ \left(
    \begin{array}{cc}
    \d_m{}^n & 0 \\
    0 & \d^m{}_n
    \end{array}
  \right) \pm \left(
    \begin{array}{cc}
    A_m{}^n & A_{mn} \\
    A^{mn} & A^m{}_n
    \end{array}
  \right) \right]
  \end{equation}
The conservation of $(j_z,j_\zb)$ produce some conditions on the generalized metric $A$ and section $(V,\ F)$ as we will show in a following section.

\subsection{Generalized metrics on supermanifolds}
Let's generalize our proposal to admit fermionic fields. This is, the target space is now a supermanifold $M$ with coordinates $(\ZZ^i)$ $i=1,\ldots,d|1,\ldots,s$ collectively denoting bosonic and fermionic ones. Then we take the action to be
  \begin{equation}\label{action}
  S[\ZZ,\PP] = \int \left\{\p_t \ZZ^i \PP_i - \frac12 (\p_\s \ZZ^i\ \PP_i) \left(
    \begin{array}{cc}
    \AA_{ij} & \AA_i{}^j \\
    \AA^i{}_j & \AA^{ij}
    \end{array}
  \right) \left(
    \begin{array}{c}
    \p_\s \ZZ^j (-)^j \\
    \PP_j
    \end{array}
  \right)\right\}
  \end{equation}
where $(\PP_i)$ is the conjugate momentum to $(\ZZ^i)$ and the matrix $\AA$ is the generalized metric on $TM\oplus T^*M$ characterizing the theory. In this super-case, the condition of $\AA$ being symmetric implies that its supermatrix blocks satisfy
  \begin{equation}
  \AA_{ij} = (-)^{ij} \AA_{ji}\ ,\quad \AA_i{}^j = (-)^{ij + j} \AA^j{}_i\ ,\quad \AA^{ij} = (-)^{ij + i +j} \AA^{ji}
  \end{equation}
while the $O(d,d)$ condition for the bosonic case generalizes to be
  \begin{equation}\label{osp-condition}
  \left(
    \begin{array}{cc}
    \AA_{ik} & \AA_i{}^k \\
    \AA^i{}_k & \AA^{ik}
    \end{array}
  \right) \left(
    \begin{array}{cc}
    0 & \d^k{}_l \\
    \d_k{}^l & 0
    \end{array}
  \right) \left(
    \begin{array}{cc}
    \AA_{lj} & \AA_l{}^j \\
    \AA^l{}_j & \AA^{lj}
    \end{array}
  \right) = \left(
    \begin{array}{cc}
    0 & \d_i{}^j \\
    \d^i{}_j & 0
    \end{array}
  \right)
  \end{equation}
This condition, as in the bosonic case, guarantees that the action (\ref{action}) has conformal invariance although not manifest. However, notice that on supermanifolds there is no notion of signature of a biliner form $\AA$, so the signature condition on $\AA$ gets dropped from the definition of generalized metric. And, to determine if the theory admits or not a second order formulation i.e. if the $(\pm 1)$-eigenspaces of $\UU_{\AA}$ can be written as the graphs of some metric and two-form tensors $(\pm\GG_{ij} + \BB_{ij})$, we need to directly analyze the invertibility of $\AA^{ij}$. 

As an example we can consider
  \begin{equation}
  \AA^{ij} = E_{a,\a}{}^i \left(
    \begin{array}{cc}
    \eta^{ab} & 0 \\
    0 & \eta^{\a\b}
    \end{array}
  \right) E_{b,\b}{}^j
  \end{equation}
where the labels $a,b,\a,\b$ are thought as flat indices, and the blocks $\eta^{ab}$, $\eta^{\a\b}$ are respectively invertible and non-invertible. Then, of course, because of a non-invertible part in $\AA^{ij}$ there is not a second order formulation but using (\ref{osp-condition}) this action could still be rewritten in a manifestly conformal invariant way
  \begin{equation}
  S\sim \int \Pi_z^a (\eta_{ab} + A_a{}^c\eta_{cb})\Pi_\zb^b + \frac12\Pi^\a_z d_{\zb \a} + \frac12\Pi^\a_\zb d_{z \a} + \frac12\Pi^\a_z A_\a{}^b \eta_{ba} \Pi^a_\zb - \frac12\Pi^\a_\zb A_\a{}^b\eta_{ba}\Pi^a_z
  \end{equation}
where $\Pi_{z,\zb}^a = \p_{z,\zb} \ZZ^i E_i{}^a$, $\Pi_{z,\zb}^\a = \p_{z,\zb} \ZZ^i E_i{}^\a$ and $(d_{z,\zb})_\a = E_\a{}^i(\PP_{t i} \pm \PP_{\s i})$. Furthermore, we can also show that when $\AA$ is written as a supermatrix (diagonal bosonic and off-diagonal fermionic blocks), equation (\ref{osp-condition}) becomes equivalent to the defining property of the orthosymplectic supergroup $OSp(d,d|2s)$. See Appendix \ref{app-ortho} for more details.

Just like in the bosonic case, we also define the components of a super current $(j_z, j_\zb)$ as
  \begin{align}
  j_z & :=  (\VV^i\ \FF_i) \left(
    \begin{array}{cc}
    (\Pc_+)_i{}^j & (\Pc_+)_{ij} \\
    (\Pc_+)^{ij} & (\Pc_+)^i{}_j
    \end{array}
  \right) \left(
    \begin{array}{c}
    \PP_j \\
    \p_\s \ZZ^j(-)^j
    \end{array}
  \right) \\
  j_\zb & := (\VV^i\ \FF_i) \left(
    \begin{array}{cc}
    (\Pc_-)_i{}^j & (\Pc_-)_{ij} \\
    (\Pc_-)^{ij} & (\Pc_-)^i{}_j
    \end{array}
  \right) \left(
    \begin{array}{c}
    \PP_j \\
    \p_\s \ZZ^j(-)^j
    \end{array}
  \right)
  \end{align}
where $\Pc_\pm$ are the projectors onto the $(\pm 1)$-eigenspaces of $\UU_\AA$ denoted by $(TM\oplus T^*M)_\pm$ and defined exactly as in (\ref{projectors})
  \begin{equation}
  \Pc_\pm := \frac12 \left[ \left(
    \begin{array}{cc}
    \d_i{}^j & 0 \\
    0 & \d^i{}_j
    \end{array}
  \right) \pm \left(
    \begin{array}{cc}
    \AA_i{}^j & \AA_{ij} \\
    \AA^{ij} & \AA^i{}_j
    \end{array}
  \right) \right]
  \end{equation}
And, again we see that this current is associated to a super vector field $\VV^i$ and super $1$-form $\FF_i$, or equivalently, to a section $(\VV\ \FF)$ of $TM\oplus T^*M$.

\subsection{Holomorphicity}
Now we can proceed to analyze the holomorphicity and anti-holomorphicity conditions on currents. First of all, we notice that when the section $(\VV\ \FF)$ associated to a current belongs to either eigenspace $(TM\oplus T^*M)_\pm$, one of the components of the current becomes zero. This is,
  \begin{equation}\label{hol}
  (\VV,\FF) \in (TM\oplus T^*M)_+ \quad \Leftrightarrow \quad j_z = \VV^i \PP_i + \FF_i \p_\s \ZZ^i(-)^i\ ,\quad j_\zb = 0
  \end{equation}
or
  \begin{equation}\label{antihol}
  (\VV,\FF) \in (TM\oplus T^*M)_- \quad \Leftrightarrow \quad j_z = 0 \ ,\quad j_\zb = \VV^i \PP_i + \FF_i \p_\s \ZZ^i (-)^i
  \end{equation}
The implication of a section satisfying any of these is that the current conservation condition, i.e. $\pb j_z + \p j_\zb = 0$, actually becomes the condition of holomorphicity $\pb j_z = 0$ or antiholomorphicity $\p j_\zb = 0$ for the current.

Finally, to obtain the symmetry conditions (and hence the holomorphiciy conditions), we compute the divergence of the current $(\pb j_z + \p j_\zb)$ using equations of motion
  \begin{align}
  \p_t \ZZ^i & = \frac12 \left[ ((-)^i \AA^i{}_j\ (-)^i \AA^{ij}) \left(
    \begin{array}{c}
    \p_\s \ZZ^j(-)^j \\
    \PP_j
    \end{array}
  \right) +  (\p_\s \ZZ^j\ \PP_j) \left(
    \begin{array}{c}
    \AA_j{}^i \\
    \AA^{ji}
    \end{array}
  \right) \right] \\
  \p_t \PP_i & = -\frac12\left[ (\p_\s \ZZ^k\ \PP_k)(-)^{ki}\left(
    \begin{array}{cc}
    \p_i \AA_{kj} & \p_i \AA_k{}^j \\
    \p_i \AA^k{}_j & \p_i \AA^{kj}
    \end{array}
  \right) \right. \left(
    \begin{array}{c}
    \p_\s \ZZ^j (-)^j \\
    \PP_j
    \end{array}
  \right) \nonumber \\
  & \qquad \left. - \p_\s \left( (\AA_{ij}\ \AA_i{}^j) \left(
    \begin{array}{c}
    \p_\s\ZZ^j(-)^j \\
    \PP_j
    \end{array}
  \right) \right) - \p_\s\left( (\p_\s\ZZ^j\ \PP_j)\left(
    \begin{array}{c}
    \AA_{ji} \\
    \AA^j{}_i
    \end{array}
  \right) \right) \right]
  \end{align}
This computation results in
  \begin{align}
  \p j_\zb + \pb j_z & = -\frac12 (\p_\s \ZZ^i\ \PP_i) \left[(-)^{i\VV} \left(
    \begin{array}{cc}
    \Lc_\VV \AA_{ij} & \Lc_\VV \AA_i{}^j \\
    \Lc_\VV \AA^i{}_j & \Lc_\VV \AA^{ij}
    \end{array}
  \right) \right.\nonumber \\
  & \quad - \left. (-)^{i\FF} \left(
    \begin{array}{cc}
    -(d\FF)_{(i|k} \AA^k{}_{|j)} & -(d\FF)_{ik} \AA^{kj} \\
    (-)^{(i+k)\FF} \AA^{ik}(d\FF)_{kj} & 0
    \end{array}
  \right) \right] \left(
    \begin{array}{c}
    \p_\s \ZZ^j(-)^j \\
    \PP_j
    \end{array}
  \right)
  \end{align}
And, since the parities of $\VV$ and $\FF$ are the same, we can extract the conditions to have a conserved current
  \begin{equation}\label{holomorphicity}
    \begin{aligned}
    (\Lc_\VV \AA)_{ij} & = -[(d\FF)_{ik} \AA^k{}_j + (-)^{ij}(d\FF)_{jk} \AA^k{}_i] \\
    (\Lc_\VV \AA)_i{}^j & = -(d\FF)_{ik}\AA^{kj} \\
    (\Lc_\VV \AA)^{ij} & = 0
    \end{aligned}
  \end{equation}
In terms of the associated $\UU_\AA$, these equation can be written as
  \begin{equation}\label{act-operator}
  \Lc_{(\VV,\ \FF)} \UU := \Lc_\VV \UU + \left[ \left(
    \begin{array}{cc}
    0 & d\FF \\
    0 & 0
    \end{array}
  \right) ,\UU \right] = 0
  \end{equation}
where $[\cdot, \cdot]$ is a commutator. The operation defined in (\ref{act-operator}) can be understood as the Lie derivative of $\UU$ along the section $(\VV,\ \FF)$, and it can be shown that it also behaves as a derivation i.e. $\Lc_{(\VV,\ \FF)} (\UU_1 \UU_2) = (\Lc_{(\VV,\ \FF)}\UU_1) \UU_2 + \UU_1 (\Lc_{(\VV,\ \FF)} \UU_2)$. 

As we already mentioned, this set of equations that determine if a section generates a symmetry becomes the holomorphicity (or antiholomorphicity) conditions when $(\VV\ \FF)$ associated to our current belongs to the subbundle $(TM\oplus T^*M)_+$ or $(TM\oplus T^*M)_-$ i.e. when the section is an eigenvector of $\UU_\AA$.

\subsection{Nilpotency}
In this part we generalize the results of \cite{Alekseev2004,Nekrasov2005} to the case where the target space is a supermanifold. This is, given currents of the form
  \begin{equation}
  j_{(\VV, \FF)} := \VV^i \PP_i + \FF_i \p_\s \ZZ^i (-)^i
  \end{equation}
we can compute their Poisson brackets to obtain
  \begin{align}
  \left\{j_{(\VV_1, \FF_1)}(\s_1), j_{(\VV_2, \FF_2)}(\s_2) \right\}_{PB} & = j_{\lb(\VV_1, \FF_1),(\VV_2,\FF_2)\rb}\d(\s_1-\s_2) \nonumber \\
  & \qquad - \langle(\VV_1,\FF_1),(\VV_2,\FF_2)\rangle\d'(\s_1-\s_2)
  \end{align}
where the brackets $\lb\cdot,\cdot\rb$ and $\langle\cdot,\cdot\rangle$ are the Dorfman bracket and canonical inner product in $TM\oplus T^*M$. Explicitly these brackets are
  \begin{align}
  \lb(\VV_1, \FF_1),(\VV_2,\FF_2)\rb & = \left([\VV_1, \VV_2],\ \Lc_{\VV_1} \FF_2 - (-)^{\VV_2 \FF_1}\i_{\VV_2}d\FF_1\right) \\
  \langle(\VV_1,\FF_1),(\VV_2,\FF_2)\rangle & = \i_{\VV_1} \FF_2 + (-)^{\VV_2 \FF_1} \i_{\VV_2}\FF_1
  \end{align}
We can also see that this Dorfman bracket is complementary to the Lie derivative of an operator defined in (\ref{act-operator}) since it can be verified that
  \begin{equation}\label{distributive}
  \lb (\VV_1,\ \FF_1), \UU(\VV_2,\ \FF_2) \rb = (\Lc_{(\VV_1,\ \FF_1)} \UU)(\VV_2,\ \FF_2 ) + \UU(\lb(\VV_1,\FF_1),(\VV_2,\FF_2)\rb)
  \end{equation}
or more compactly, if we denote $\SS = (\VV,\ \FF)$, $\SS\cdot \UU = \Lc_\SS \UU$ and $\SS_1\cdot\SS_2 = \lb \SS_1,\SS_2\rb$, the previous equation becomes $\SS_1\cdot(\UU(\SS_2)) = (\SS_1\cdot\UU)(\SS_2) + \UU(\SS_1\cdot\SS_2)$. 

The Dorfman bracket give us another way of characterizing symmetries. It was shown in \cite{Severa2017} that for a sigma model $S=\int (G+B)\pb x^m \p x^n$, a section $(V, F)$ generates a symmetry (i.e. $\Lc_V G =0$ and $\Lc_V G = dF$) if and only if $\lb (V, F), \cdot\rb$ preserves the eigenspaces $graph(\pm G + B)$. Using (\ref{distributive}) this result can be generalized to the following. A section $(\VV, \FF)$ generates a symmetry (i.e. $\Lc_{(\VV,\FF)} \UU_\AA = 0$) of the action (\ref{action}) if and only if the operation $\lb(\VV,\ \FF), \cdot\rb$ preserves the eigenspaces $(TM\oplus T^*M)_\pm$.
  \begin{equation}
  \Lc_{(V,F)} \UU = 0 \quad \Leftrightarrow \quad \lb (V,\ F),\ \cdot\ \rb\ \text{preserves}\ (TM\oplus T^*M)_\pm
  \end{equation}

Now we proceed to investigate the nilpotency of a fermionic current $j_{(\QQ, \FF)}$ associated to a section $(\QQ,\FF)$. Taking its Poisson bracket with itself we get
  \begin{equation}
  \{j_{(\QQ,\FF)}, j_{(\QQ,\FF)}\}_{PB} = 0\quad \Longrightarrow \quad j_{\lb(\QQ,\FF),(\QQ,\FF)\rb} = 0
  \end{equation}
since the inner product $\langle (\QQ,\FF),(\QQ,\FF)\rangle$ is trivially zero. This condition implies that the Dorfman bracket needs to be zero $\lb (\QQ,\FF),(\QQ,\FF) \rb = 0$ which in terms of the vector and $1$-form components reads
  \begin{equation}\label{nilpotency}
  [\QQ,\QQ] = 0\ ,\quad \Lc_\QQ \FF + \i_\QQ d\FF = 0
  \end{equation}
Notice that we are requiring the nilpotency of the currents since for fermionic ones, this is equivalent to nilpotency of its charge $Q = \int d\s j_{(\QQ, \FF)}(\s)$.

\section{Pure spinor string}
In this section we will show that the Pure Spinor action in curved background admits a Hamiltonian form, and that its matrix $\AA_{PS}$ satisfies the properties of a generalized metric. Also we will see that $\AA_{PS}$ possesses two eigenvectors that generate the BRST currents of the theory.

The Pure Spinor action in a curved background \cite{Berkovits2001} is written in terms of worldsheet matter fields $(Z^M) = (X^m,\t^\m,\th^\mh)$ and ghost fields $(\l^\a,\lh^\ah)$ and momenta $(\o_\a,\oh_\ah)$. Plus, the ghost variables satisfy the pure spinor constraint $(\l \g^a \l) = (\lh\g^a\lh) = 0$
  \begin{align}
  S_{PS} = \int d^2z &\ \frac12 \pb Z^N \p Z^M(G_{MN} + B_{MN}) + \pb Z^M E^\a_M d_\a + \p Z^M E^\ah_M {\widehat d}_\ah \nonumber \\
  & \quad + \pb Z^M \O_{M\a}{}^\b \l^\a \o_\b + \p Z^M \widehat\O_{M\ah}{}^\bh \lh^\a\oh_\bh \nonumber \\
  & \quad + P^{\a\ah} d_\a {\widehat d}_{\ah} + C_\a^{\b\bh} \l^\a \o_\b {\widehat d}_\bh + {\widehat C}_\ah^{\bh\b} \lh^\ah \oh_\bh {d}_\b \nonumber \\
  & \quad + S_{\a\ah}^{\b\bh} \l^\a \o_\b \lh^\ah \oh_\bh + \o_\a\pb \l^\a + \oh_\ah \p \lh^\ah \label{ps-action}
  \end{align}
The background fields $E_M{}^A$, $G_{MN}$ and $B_{MN}$ are the vielbein, super-metric and two-form, $\O_{M\a}{}^\b$ and $\Oh_{M\ah}{}^\b$ are the spin connections, the fields $C^{\a\ah}_\b$, $\widehat C^{\ah\a}_\bh$ are related to the gravitino and dilatino field strengths while $P^{\a\ah}$ encodes the Ramond-Ramond fields in its gamma matrix expansion and $S^{\a\ah}_{\b\bh}$ is related to the curvature of the manifold.

The other elements of the Pure Spinor formalism are BRST currents given by $j_B = \l^\a d_\a$ and $\widetilde{j}_B = \lh^\ah {\widehat d}_\ah$. To transform it into its Hamiltonian form, we first need to assume that $P^{\a\ah}$ is invertible%
\footnote{
A backgrounds satisfying this requirement is AdS \cite{Berkovits2010}.
}%
, so we can solve for $d_\a$ and ${\widehat d}_\ah$ using their equations of motion. The EOM for $d_\a$ and $\widehat d_\ah$ are
  \begin{align}
  d_\a & = \left( \p Z^M E_M{}^\ah + (\l C\o)^\ah \right) P^{-1}_{\ah\a} \\
  \hat d_\ah & = - P^{-1}_{\ah\a} \left( \pb Z^M E_M{}^\a + (\lh \widehat C \o)^\a\right)
  \end{align}
and we can replace these expressions back into the action and BRST currents. This will result in the field $P^{-1}_{\ah\a}$ appearing everywhere; however, when we use the Legendre transform on the matter fields $(Z^M)$ to compute the Hamiltonian form of the action (i.e. solving for $\p_t Z^M$ in the relation $P_M = \frac{\d L }{\d \p_t Z^M}$) all terms containing $P^{-1}_{\ah\a}$ vanish in the action
  \begin{equation}\label{raw-action}
  \begin{aligned}
  S_{PS} = \int\p_t Z^M P_M + \o_\a & \pb \l^\a + \oh_\ah \p \lh^\ah \\
  - \frac12 \p_\s Z^N \p_\s Z^M \Big( G_{MN} - B_{ML}A^{LP} B_{PN} & - E_{(M|}{}^\ah E_\ah{}^P B_{P|N)} + E_{(M|}{}^\a E_\a{}^P B_{P|N)}\Big) \\
  -\frac12 (P_M - \l\O_M\o - \lh\Oh_M\oh) & A^{MN} (P_N - \l\O_N\o - \lh\Oh_N\oh)  \\
  + (\l C\o^\ah E_\ah{}^M + \lh{\widehat C}\oh^\a E_\a{}^M)(P_M - & \l\O_M\o - \lh\Oh_M\oh) + (\l\lh S)^{\a\ah}\o_\a \oh_\ah \\
  + \p_\s Z^N \Big( B_{NP}A^{PM} + E_N{}^\ah E_\ah{}^M - & E_N{}^\a E_\a{}^M \Big) (P_M - \l\O_M\o - \lh\Oh_M\oh) \\
  + \p_\s Z^N (\l C\o)^\ah E_\ah{}^M B_{MN} + \p_\s Z^N (\lh{\widehat C}\oh)^\a E_\a{}^M & B_{MN} + \p_\s Z^M (\lh\Oh_M\oh) - \p_\s Z^M (\l\O_M\o)  
  \end{aligned}
  \end{equation}
where $A^{MN} = \left( (-)^M E_a{}^M E_b{}^N\eta^{ab} + E_\a{}^M E_\ah{}^N P^{\a\ah} - E_\ah{}^M E_\a{}^N P^{\a\ah}\right)$. The same occurs when the BRST currents are computed, obtaining expressions free from $P^{-1}_{\ah\a}$
  \begin{align}
  j_B & = \l^\a E_\a{}^M P_M - \l^\a E_\a{}^M(\l\O_M\o) - \l^\a E_\a{}^M(\lh\Oh_M\oh) - \p_\s Z^M B_{MN} \l^\a E_\a{}^N \label{brst-z}\\
  {\widetilde j}_B & = \lh^\ah E_\ah{}^M P_M - \lh^\ah E_\ah{}^M(\l\O_M\o) - \lh^\ah E_\ah{}^M(\lh\Oh_M\oh) - \p_\s Z^M B_{MN} \lh^\ah E_\ah{}^N \label{brst-zb}
  \end{align}
Notice that these expressions for the BRST currents are exactly the same as the ones appearing in \cite{Berkovits2001}.

This action (\ref{raw-action}) can be put in the form of (\ref{action}) as we will see in next section, and more importantly we claim that the condition on the invertibility of $P^{\a\ah}$ can be dropped since even the most degenerate case such as Pure Spinor in flat space (where $P^{\a\ah}=0$) is described by (\ref{raw-action}).

\subsection{Generalized metric and BRST sections for the Pure Spinor string}
The Pure Spinor action in Hamiltonian form (\ref{raw-action}) can be written as in (\ref{action}) where there is a generalized metric $\AA$ containing the all background fields that characterize the theory. We consider the target space to be parametrized by
  \begin{equation}
  ({\mathbb Z}^i) = (Z^M, \l^\a, \lh^\ah)
  \end{equation}
matter and ghost fields, while their momentum variables are $({\mathbb P}_i) = (P_M, \o_\a, \oh_\ah)$. Here the bosonic coordinates are $(x^m, \l^\a, \lh^\ah)$ and the fermionic ones $(\t^\m,\th^\mh)$. With these considerations, the generalized metric in the Pure Spinor case $\AA_{PS}$ can be read to be
  \begin{equation}\label{ps-met}
  \AA_{PS} = \left(
    \begin{array}{ccc|ccc}
    \AA_{MN}    & 0          & 0             & \AA_M{}^N   & \AA_M{}^\b   & \AA_M{}^\bh  \\
    0           & 0          & 0             & 0           & \d_\a{}^\b   & 0            \\
    0           & 0          & 0             & 0           & 0            & -\d_\ah{}^\bh \\
    \hline
    \AA^M{}_N   & 0          & 0             & \AA^{MN}    & \AA^{M\b}    & \AA^{M\bh}   \\
    \AA^\a{}_N  & \d^\a{}_\b & 0             & \AA^{\a N}  & \AA^{\a\b}   & \AA^{\a\bh} \\
    \AA^\ah{}_N & 0          & -\d^\ah{}_\bh & \AA^{\ah N} & \AA^{\ah \b} & \AA^{\ah\bh}
    \end{array}
  \right)
  \end{equation}
where the expressions for each component block are left to the Appendix \ref{app}.

Next we would like to have some fermionic sections $(\QQ\ \FF)$ and $(\QQt\ \FFt)$ that reduce to the BRST currents (\ref{brst-z}-\ref{brst-zb}) after showing that they are eigenvectors of $\UU_{(\AA_{PS})}$. By looking at the BRST currents we see that there is a natural guess, and they are given by  $(\QQ^i , \FF_i ) = (Q^M, Q^\a, Q^\ah;\ F_M, F_\a, F_\ah)$ where
  \begin{equation}\label{brst-sect1}
  (\QQ^i) = \left(
    \begin{array}{c}
    Q^M \\
    Q^\a \\
    Q^\ah
    \end{array}
  \right) = \left(
    \begin{array}{c}
    \l^\a E_\a{}^M \\
    - \l^\b E_\b{}^M (\l \O_M)^\a \\
    -\l^\b E_\b{}^M (\lh \Oh_M)^\ah
    \end{array}
  \right)\ ,\quad (\FF_i) = \left(
    \begin{array}{c}
    F_M \\
    F_\a \\
    F_\ah
    \end{array}
  \right) = \left(
    \begin{array}{c}
    Q^N B_{NM} \\
    0 \\
    0
    \end{array}
  \right)
  \end{equation}
and similarly for $({\widetilde \QQ}^i , {\widetilde\FF}_i ) = ({\widetilde Q}^M, {\widetilde Q}^\a, {\widetilde Q}^\ah;\ {\widetilde F}_M, {\widetilde F}_\a, {\widetilde F}_\ah)$, we have
  \begin{equation}\label{brst-sect2}
  (\QQt^i) = \left(
    \begin{array}{c}
    \widetilde Q^M \\
    \widetilde Q^\a \\
    \widetilde Q^\ah
    \end{array}
  \right) = \left(
    \begin{array}{c}
    \lh^\ah E_\ah{}^M \\
    - \lh^\bh E_\bh{}^M (\l \O_M)^\a \\
    -\lh^\bh E_\bh{}^M (\lh \Oh_M)^\ah
    \end{array}
  \right)\ ,\quad (\FFt_i) = \left(
    \begin{array}{c}
    \widetilde F_M \\
    \widetilde F_\a \\
    \widetilde F_\ah
    \end{array}
  \right) = \left(
    \begin{array}{c}
    \widetilde Q^N B_{NM} \\
    0 \\
    0
    \end{array}
  \right)
  \end{equation}

Next we state our main results concerning the description of the pure spinor sigma model in terms of generalized geometry. These results can be proven with a lengthy but straightforward computation
  \begin{itemize}
  \item The generalized metric $\AA_{PS}$ corresponding to the Pure Spinor action satisfies the property (\ref{osp-condition})
    \begin{equation}
    \AA_{PS} \left(
      \begin{array}{cc}
      0 & 1 \\
      1 & 0
      \end{array}
    \right) \AA_{PS} = \left(
      \begin{array}{cc}
      0 & 1 \\
      1 & 0
      \end{array}
    \right)
    \end{equation}
    or equivalently, its associated $\UU_{(\AA_{PS})}$ satisfies $\UU_{(\AA_{PS})}^2 = 1$. This means that the Pure Spinor action $S_{PS}$ in Hamiltonian form still possesses $SO(1,1)$ worldsheet symmetry, and it can be rewritten with manifest conformal invariance which is the original sigma model (\ref{ps-action}). Plus, we also have that when $\AA_{PS}$ is reordered as a super-matrix, it belongs to the orthosymplectic supergroup $OSp(d,d|2s)$.
    
    \item The sections $(\QQ^i\ \FF_i)$ and $(\QQt^i\ \FFt_i)$ that we defined in (\ref{brst-sect1}) and (\ref{brst-sect2}) satisfy the property of being eigenvectors of $\mathbb U_{\AA}$ with eigenvalues $(+1)$ and $(-1)$ respectively. In turn this implies that the components of the currents associated to these sections coincide with the BRST currents
      \begin{align}
      (\QQ\ \FF) & \quad \longrightarrow \quad j_z = j_B\ ,\quad j_\zb = 0 \\
      (\QQt\ \FFt) & \quad \longrightarrow \quad j_z = 0\ ,\quad j_\zb = \widetilde j_B
      \end{align}
    Then, as was discussed, the symmetry conditions for each case become holomorphicity and antiholomorphicity conditions for the BRST currents. These equations will ultimately imply the Type II SUGRA constraints (See section \ref{sect-sugra}).
    
    \item The target space of the pure spinor sigma model is a graded supermanifold (grading given by ghost number). This fact implies that there cannot exist tensors $(\GG_{ij})_{i,j=M,\a}$ and $(\BB_{ij})_{i,j=M,\a}$ such that $(TM\oplus T^*M)_\pm$ are solved to be the graph of $(\pm\GG+\BB)$. In other words, the pure spinor sigma model does not admit a second order formulation (See section \ref{grading} for details).
  \end{itemize}

\subsection{Type II SUGRA constraints}\label{sect-sugra}
We now show that for the Pure Spinor case the conditions of nilpotency (\ref{nilpotency}) and holomorphicity (\ref{holomorphicity}) that we found actually imply the Type II SUGRA constraints which of course agrees with the results of \cite{Berkovits2001}.

The constraints of nilpotency (\ref{nilpotency}) in the Pure Spinor case can be simplified further. To avoid a cumbersome computation let's define a $2$-form field
  \begin{equation}
  [\BB_{ij}] := \left(
    \begin{array}{ccc}
    B_{MN} & 0 & 0 \\
    0 & 0 & 0 \\
    0 & 0 & 0
    \end{array}
  \right)
  \end{equation}
with field strength
  \begin{equation}
  {\mathbb H}_{ijk} = \left\{
    \begin{array}{ll}
    H_{MNP} := (dB)_{MNP} & ; i=M,\ j=N,\ k=P \\
    0 & ;\ \text{any other case}
    \end{array}
  \right.
  \end{equation}
This means that we can write the $1$-form simply as $\FF = \i_\QQ \BB$, and then the second equation in (\ref{nilpotency}) becomes
  \begin{align}
  \Lc_\QQ (\i_\QQ\BB) + \i_\QQ d(\i_\QQ\BB) & = \i_{[\QQ,\QQ]} \BB - (\i_\QQ \Lc_\QQ\BB) + \i_\QQ d(\i_\QQ\BB) \nonumber \\
  & = \i_{[\QQ,\QQ]}\BB - \i_\QQ \i_\QQ d\BB
  \end{align}
where in the last line we used Cartan's formula $\Lc_\QQ = \i_\QQ d + d \i_\QQ$. Thus, together with the Lie bracket $[\QQ, \QQ] = 0$, nilpotency of the BRST current $j_B$ implies the following constraints
  \begin{equation}\label{nil1}
  [\QQ,\QQ] = 0 \ ,\quad \i_\QQ\i_\QQ \HH = 0
  \end{equation}
In a completely analogous way we obtain the rest of constraints for the relations $\{{\widetilde j}_B, {\widetilde j}_B\} = 0$ and $\{j_B, {\widetilde j}_B\} = 0$
  \begin{align}
  [\QQt,\QQt] = 0\ ,\quad \i_\QQt\i_\QQt \HH = 0 \label{nil2} \\
  [\QQ, \QQt] = 0\ ,\quad \i_\QQ\i_\QQt \HH = 0 \label{nil3}
  \end{align}
Then, by actually computing the Lie brackets we arrive at
  \begin{align}
  \l^\a\l^\b T_{\a\b}{}^C = \l^\a\l^\b\l^\g R_{\a\b\g}{}^\s = \l^\a\l^\b \widehat R_{\a\b\gh}{}^\sh = 0 & \ ,\quad \l^\a\l^\b H_{\a\b C} = 0 \label{sugra-nil1} \\
  \lh^\ah\lh^\bh T_{\ah\bh}{}^C = \lh^\ah\lh^\bh R_{\ah\bh\g}{}^\s = \lh^\ah\lh^\bh\lh^\gh \widehat R_{\ah\bh\gh}{}^\sh = 0 & \ ,\quad \lh^\ah\lh^\bh H_{\ah\bh C} = 0 \label{sugra-nil2}\\
  T_{\a\bh}{}^C = \l^\a\l^\g R_{\a\bh\g}{}^\s = \lh^\ah\lh^\gh \widehat R_{\ah\b\gh}{}^\sh = 0 & \ ,\quad H_{\a\bh C} = 0 \label{sugra-nil3}
  \end{align}

Finally, using (\ref{holomorphicity}) in the Pure Spinor case, we have that the holomorphicity conditions $\pb j_B = 0$ become
  \begin{equation}\label{hol1}
  (\Lc_\QQ \AA)_{ij} + (d\i_\QQ \BB)_{(i|k} \AA^k{}_{|j)} = 0\ ,\quad (\Lc_\QQ \AA)_i{}^j + (d\i_\QQ \BB)_{ik} \AA^{kj} = 0 \ ,\quad (\Lc_\QQ \AA)^{ij} = 0
  \end{equation}
and the antiholomorphicity ones $\p {\widetilde j}_B = 0$ become
  \begin{equation}\label{hol2}
  (\Lc_\QQt \AA)_{ij} + (d\i_\QQt \BB)_{(i|k} \AA^k{}_{|j)} = 0\ ,\quad (\Lc_\QQt \AA)_i{}^j + (d\i_\QQt \BB)_{ik} \AA^{kj} = 0 \ ,\quad (\Lc_\QQt \AA)^{ij} = 0
  \end{equation}
Using the explicit form of the generalized metric and sections, we can compute these equations. Schematically, the Lie derivative $\Lc_{(\QQ,\FF)}$ will act on the vielbein $E_A{}^M$, spin connection $\O_{M\a}{}^\b$ and two-form $B_{MN}$ producing the torsion $\l^\g T_{\g A}{}^M$, curvature tensor $\l^\g R_{\g M\a}{}^\b$ and the three-form field strength $\l^\g H_{\g M N}$. For the first set of equations (\ref{hol1}) we get
  \begin{equation}\label{sugra-hol1}
    \begin{aligned}
    & T_{\a(bc)} = H_{\a bc} = H_{\a \bh \g} = T_{\a\b c} + H_{\a\b c} = T_{\a\bh c} - H_{\a\bh c} = 0  \\
    & T_{\a b}{}^\g - T_{\a\gh b}P^{\g\gh} = T_{\a b}{}^\gh + T_{\a\g b}P^{\g\gh} = T_{\a\b}{}^\gh + \frac12 H_{\a\b\g}P^{\g\gh} = T_{\a\bh}{}^\g = 0 \\
    & \N_\a P^{\b\bh} + T_{\a\rho}{}^\b P^{\rho\bh} - C^{\b \bh}_\a = \widehat R_{a\b \gh}{}^\sh - T_{\b\rho a}\widehat C^{\sh\rho}_\gh = \widehat R_{\a\b\gh}{}^\sh + \frac12 H_{\a\b\bullet} \widehat C^{\sh \bullet}_\gh = 0 \\
    & - S_{\a\ah}^{\b\bh} + \widehat R_{\a\gh\ah}{}^\bh P^{\b\gh} + \N_\a \widehat C^{\bh\b}_\ah + T_{\a\g}{}^\b \widehat C^{\bh\g}_\ah = 0 \\
    & \l^\a \l^\b \left( R_{a\a\b}{}^\g - T_{\a\bh a} C^{\g \bh}_\b \right) = \l^\a\l^\b R_{\bh \a\b}{}^\g = 0 \\
    & \l^\a \l^\b \left( \N_\a C_\b^{\g\gh} - R_{\a\s\b}{}^\g P^{\s\gh}\right) = \l^\a\l^\b \left( \N_\a S_{\b\bh}^{\g\gh} - \widehat R_{\a\bullet\bh}{}^\gh C_\b^{\g\bullet} - R_{\a\bullet \b}{}^{\g} \widehat C^{\gh\bullet}_\bh \right) = 0
    \end{aligned}
  \end{equation}
and for the other (\ref{hol2})
  \begin{equation}\label{sugra-hol2}
    \begin{aligned}
    & T_{\ah(bc)} = H_{\ah bc} = H_{\ah \b \gh} = T_{\ah\bh c} - H_{\ah\bh c} = T_{\ah\b c} + H_{\ah\b c} = 0  \\
    & T_{\ah b}{}^\gh - T_{\ah\g b}P^{\g\gh} = T_{\ah b}{}^\g + T_{\ah\gh b}P^{\g\gh} = T_{\ah\bh}{}^\g + \frac12 H_{\ah\bh\gh}P^{\g\gh} = T_{\ah\b}{}^\gh = 0 \\
    & \N_\ah P^{\b\bh} + T_{\ah\widehat\rho}{}^\bh P^{\b\widehat\rho} - \widehat C^{\bh \b}_\ah = R_{a\bh \g}{}^\s - T_{\bh\widehat\rho a} C^{\s\widehat\rho}_\g = R_{\ah\bh\g}{}^\s + \frac12 H_{\ah\bh\bullet} C^{\s \bullet}_\g = 0 \\
    & - S_{\a\ah}^{\b\bh} + R_{\ah\g\a}{}^\b P^{\g\bh} + \N_\ah C^{\b\bh}_\a + T_{\ah\gh}{}^\bh C^{\b\gh}_\a = 0 \\
    & \lh^\ah \lh^\bh \left( \widehat R_{a\ah\bh}{}^\gh - T_{\ah\b a} \widehat C^{\b\gh}_\bh \right) = \lh^\ah\lh^\bh \widehat R_{\b \ah\bh}{}^\gh = 0 \\
    & \lh^\ah \lh^\bh \left( \N_\ah \widehat C_\bh^{\gh\g} - \widehat R_{\ah\sh\bh}{}^\gh P^{\g\sh}\right) = \lh^\ah\lh^\bh \left( \N_\ah S_{\b\bh}^{\g\gh} - R_{\ah\bullet\b}{}^\g \widehat C_\bh^{\gh\bullet} - \widehat R_{\ah\bullet \bh}{}^{\gh} C^{\g\bullet}_\b \right) = 0
    \end{aligned}
  \end{equation}
These set of equations are the same%
\footnote{
The difference in relative signs come from our use of different conventions for the contraction of vectors and tensors.
}
 as in \cite{Berkovits2001}, and it was proven there that they imply the correct Type II SUGRA constraints.

\subsection{Flat space}\label{flat}
So far we showed how it all works in the most general setting but it may help to do some explicit computations in the simplest case.

The flat space background corresponds to setting the vielbein fields to
  \begin{equation}
  [E_M{}^A] = \left(
    \begin{array}{ccc}
    \d_m^a & 0 & 0 \\
    -\frac12 (\g^a\t)_\m & \d_\m^\a & 0 \\
    -\frac12 (\g^a\th)_\mh & 0 & \d_\mh^\ah
    \end{array}
  \right)\ ,\quad [E_A{}^M] = \left(
    \begin{array}{ccc}
    \d_a^m & 0 & 0 \\
    \frac12 (\g^m\t)_\a & \d_\a^\m & 0 \\
    \frac12 (\g^m\th)_\ah & 0 & \d_\ah^\mh
    \end{array}
  \right)
  \end{equation}
the $B$-field to
  \begin{equation}
  [B_{AB}] = -\frac12 \left(
    \begin{array}{ccc}
    0 & -(\g_a\t)_\b & (\g_a\th)_\bh \\
    (\g_b\t)_\a & 0 & \frac12 (\g^a\t)_\a(\g_a\th)_\bh \\
    -(\g_b\th)_\ah & -\frac12 (\g^a\th)_\ah(\g_a\t)_\b & 0
    \end{array}
  \right)
  \end{equation}
and the rest of background fields $\O_{M\a}{}^\b,\Oh_{M\ah}{}^\bh, P^{\a\ah}, C_\a^{\b\bh} , \widehat C_\ah^{\bh \b}$ to zero. Then the block matrix $\AA$ simplifies greatly. In $[\AA^{ij}]$ and $[\AA_{ij}]$ the only non-zero blocks are $\AA^{MN}$ and $\AA_{MN}$ respectively, where
  \begin{equation}
  [\AA^{MN}] = \left(
    \begin{array}{ccc}
    \eta^{mn} & 0 & 0 \\
    0 & 0 & 0 \\
    0 & 0 & 0
    \end{array}
  \right)\ , \qquad [\AA_{MN}] = \left(
    \begin{array}{ccc}
    \eta_{mn} & 0 & 0 \\
    0 & 0 & 0 \\
    0 & 0 & 0
    \end{array}
  \right)
  \end{equation}
Likewise, $[\AA_i{}^j]$ reduces to
  \begin{equation}
  [\AA_i{}^j] = \left(
    \begin{array}{ccc}
    \AA_M{}^N & 0 & 0 \\
    0 & \d_\a{}^\b & 0 \\
    0 & 0 & -\d_\ah{}^\bh
    \end{array}
  \right)\quad\text{where}\quad [\AA_M{}^N] = \left(
    \begin{array}{ccc}
    0 & 0 & 0 \\
    0 & \d_\m{}^\n & 0 \\
    0 & 0 & -\d_\mh{}^\nh
    \end{array}
  \right)
  \end{equation}
This implies that the action takes the following form
  \begin{align}
  S & = \int \p_t x^m P_m + \p_t \t^\a p_\a + \p_t \th^\ah \widehat p_\ah + \p_t \l^\a \o_\a + \p_t\lh^\ah \oh_\ah  \nonumber \\
  & \qquad - \left( \frac12 P_m \eta^{mn} P_n + \frac12 \p_\s x^n\p_\s x^m \eta_{mn} - (\p_\s\th^\ah \widehat p_\ah - \p_\s\t^\a p_\a) - (\p_\s\lh^\ah \oh_\ah - \p_\s\l^\a \o_\a)\right) \nonumber \\
  & = \int (\p_t x^m P_m -\frac12 P_m P^m -\frac12 \p_\s x^n\p_\s x^m\eta_{mn}) + \pb\t^\a p_\a + \p\th^\ah \widehat p_\ah + \pb\l^\a \o_\a + \p\lh^\ah \oh_\ah \label{flat-action}
  \end{align}
and after solving for $P_m$ using its equations of motion we end up with the Pure Spinor string in flat space
  \begin{equation}
  S = \int \frac12 \pb x^m \p x_m + \pb\t^\a p_\a + \p\th^\ah \widehat p_\ah + \pb\l^\a \o_\a + \p\lh^\ah \oh_\ah
  \end{equation}

The eigenspaces $(TM\oplus T^*M)_\pm$ in flat space can be easily computed since the operator $\UU_{\AA_{flat}}$ is simply
  \begin{equation}
  \UU_{\AA_{flat}} = \left(
    \begin{array}{ccccc|ccccc}
    0 & & & & & \eta_{mn} & & & &\\
    & \d^\m_\n & & & & &  0 & & & \\
    & & -\d^\mh_\nh & & & & & 0 & & \\
    & & & \d^\a_\b & & & & & 0 & \\
    & & & & -\d^\ah_\bh & & & & & 0 \\
    \hline
    \eta^{mn} & & & & & 0 & & & & \\
    & 0 & & & & &  \d^\m_\n & & & \\
    & & 0 & & & & & -\d^\mh_\nh & & \\
    & & & 0 & & & & & \d^\a_\b & \\
    & & & & 0 & & & & & -\d^\ah_\bh 
    \end{array}
  \right)
  \end{equation}
Then, the $(+1)$-eigenvectors are
  \begin{equation}
  (\VV, \FF) = \left( \VV^m \frac{\p}{\p x^m} + \VV^\m \frac{\p}{\p\t^\m} + \VV^\a \frac{\p}{\p\l^\a},\ \VV^n \eta_{nm} \p_\s x^m + \FF_\m \p_\s \t^\m + \FF_\a \p_\s \l^\a \right)
  \end{equation}
and the $(-1)$-eigenvectors are
  \begin{equation}
  (\widetilde\VV, \FFt) = \left( \widetilde\VV^m \frac{\p}{\p x^m} + \widetilde\VV^\mh \frac{\p}{\p\th^\mh} + \widetilde\VV^\ah \frac{\p}{\p\lh^\ah},\ -\widetilde\VV^n \eta_{nm} \p_\s x^m + \FFt_\mh \p_\s \th^\mh + \FFt_\ah \p_\s \lh^\ah \right)
  \end{equation}
for $\VV^{m,\m,\a}$, $\widetilde\VV^{m,\mh,\ah}$, $\FF_{\m,\a}$, $\FFt_{\mh,\ah}$ functions of $(x^m, \t^\m, \th^\mh,\l^\a, \lh^\ah)$. Thus, we can see both eigen\-spaces have the same dimension. Furthermore, it becomes clear that these eigen\-spaces $(TM\oplus T^*M)_\pm$ have non-trivial intersections with $TM$ and $T^*M$
  \begin{equation}
  \begin{aligned}
  (TM\oplus T^*M)_+ \cap TM & = \left\{ \VV^\m \frac{\p}{\p\t^\m} + \VV^\a \frac{\p}{\p\l^\a} \right\} \\
  (TM\oplus T^*M)_+ \cap T^*M & = \left\{ \FF_\m \p_\s \t^\m + \FF_\a \p_\s \l^\a \right\} \\
  (TM\oplus T^*M)_- \cap TM & = \left\{ \widetilde\VV^\mh \frac{\p}{\p\th^\mh} + \widetilde\VV^\ah \frac{\p}{\p\lh^\ah} \right\} \\
  (TM\oplus T^*M)_- \cap T^*M & = \left\{ \FFt_\mh \p_\s \th^\mh + \FFt_\ah \p_\s \lh^\ah \right\}
  \end{aligned}
  \end{equation}
In the generic case, these intersections reduce but never become trivial. We will see in the next section that because of the ghost number grading on the target space, the intersections between $(TM\oplus T^*M)_\pm$ and $T^*M$ may be lifted to be zero but the directions $\frac{\p}{\p \l^\a}$ and $\frac{\p}{\p\lh^\ah}$ will always be in $(TM\oplus T^*M)_+ \cap TM$ and $(TM\oplus T^*M)_-\cap TM$ respectively.

Special cases of sections belonging to these eigenspaces are the ones that generate the BRST currents that for flat space become
  \begin{align}
  j_B & = \frac12(\l\g^m\t) P_m + \l^\a p_\a + \frac12 (\l\g_m\t)\p_\s x^m + \frac14 (\l\g^m\t)(\g_m\t)_\m \p_\s \t^\m \\
  {\widetilde j}_B & = \frac12(\lh\g^m\th) P_m + \lh^\ah {\widehat p}_\ah + \frac12 (\lh\g_m\th)\p_\s x^m + \frac14 (\lh\g^m\th)(\g_m\th)_\mh \p_\s \th^\mh
  \end{align}
and after solving for $P_m$ in (\ref{flat-action}) will take their usual form in flat space
  \begin{align}
  j_B & = \l^\a \left( p_\a  + \frac12 (\g_m\t)_\a \p x^m + \frac18 (\g^m\t)_\a(\g_m\t)_\m \p \t^\m \right) \\
  j_B & = \lh^\ah \left( {\widehat p}_\ah  + \frac12 (\g_m\th)_\ah \pb x^m + \frac18 (\g^m\th)_\ah(\g_m\th)_\mh \pb \th^\mh \right)
  \end{align}

Since the flat case is the simplest solution to Type II supergravity equations, we can verify that the conditions of holomorphicity and nilpotency are satisfied. We take as an example the holomorphic case, and compute the Lie derivatives w.r.t. $\QQ$. In the case of $[\Lc_\QQ \AA_{ij}]$ the only non-zero components are
  \begin{equation}
  [(\Lc_\QQ\AA)_{MN}] = \left(
    \begin{array}{ccc}
    0 & -\frac12 (\l\g_m)_\n & 0 \\
    -\frac12 (\l\g_n)_\m & 0 & 0 \\
    0 & 0 & 0
    \end{array}
  \right)\ , \quad [(\Lc_\QQ \AA)_{\a N}] = \left( \frac12 (\g_n\t)_\a,\ 0,\ 0 \right)
  \end{equation}
and $(\Lc_\QQ\AA)_{M\b} = (\Lc_\QQ\AA)_{\b M}$. For the matrix $[(\Lc_\QQ\AA)_i{}^j]$, the only non-zero components are
  \begin{equation}
  [(\Lc_\QQ\AA)_M{}^N] = \left(
    \begin{array}{ccc}
    0 & 0 & 0 \\
    \frac12(\l\g^n)_\m & 0 & 0 \\
    0 & 0 & 0
    \end{array}
  \right)\ ,\quad [(\Lc_\QQ\AA)_\a{}^N] = \left( -\frac12(\g^n\t)_\a,\ 0,\ 0 \right)
  \end{equation}
And, for $[(\Lc_\QQ\AA)^{ij}]$ we encounter that all of its components vanish. Furthermore, to verify the holomorphicity equations we still to compute $[(d\FF)_{ij}]$ which has non-zero components
  \begin{align}
  [(d\FF)_{MN}] & = \left(
    \begin{array}{ccc}
    0 & \frac12 (\l\g_m)_\n & 0 \\
    -\frac12 (\l\g_n)_\m & \frac14 (\l\g_m)_{(\m}(\g^m\t)_{\n)} -\frac12 (\l\g_m\t)\g^m_{\m\n} & 0 \\
    0 & 0 & 0
    \end{array}
  \right) \\
  [(d\FF)_{\a N}] & = \left(\frac12 (\g_n\t)_\a,\ -\frac14(\g^m\t)_\a(\g_m\t)_\n,\ 0\right)
  \end{align}
and $(d\FF)_{M\b} = - (d\FF)_{\b M}$. With all this, it is simple to verify that the holomorphicity equations (\ref{hol1}) are satisfied in the flat space case i.e. $\pb j_B = 0$ as expected.

Finally, to verify the nilpotency conditions (\ref{nil1}), we compute the Lie bracket $[\QQ,\QQ]$ where the only non-zero component turns out to be
  \begin{equation}\label{eq1}
  ([\QQ,\QQ])^m = \l\g^m \l
  \end{equation}
Meanwhile, the field strength $H_{MNP}$ in flat space is 
  \begin{equation}
    \begin{aligned}
    H_{m\n\rho} = (\g_m)_{\n\rho}\ &,\quad H_{m\nh{\widehat\rho}} = -(\g_m)_{\nh\widehat\rho} \ , \\
    H_{\m\n\widehat\rho} = -\frac12 (\g_a)_{\m\n}(\g^a\th)_{\widehat\rho}\ &,\quad H_{\mh\nh\rho} = \frac12 (\g_a)_{\mh\nh}(\g^a\t)_\rho
    \end{aligned}  
  \end{equation}
and zero for any other mix of indices. This implies that the only non-zero components of $(\i_\QQ\i_\QQ \HH)$ are $(\i_Q\i_Q H)_M = (-)^P Q^P Q^N H_{NPM}$ that explicitly become
  \begin{equation}\label{eq2}
    \begin{aligned}
    (\i_Q \i_Q H)_m & = - (\l\g_m\l) \\
    (\i_Q \i_Q H)_\m & = -(\l\g^m\t)(\l\g_m)_\m \\
    (\i_Q \i_Q H)_\mh & = \frac12 (\l\g_m\l)(\g^m\th)_\mh
    \end{aligned}
  \end{equation}
Then the values in both equations (\ref{eq1}) and (\ref{eq2}) vanish because of the pure spinor constraint $(\l\g^m\l) = 0$ and the Fierz identity $(\g_m)_{\s(\m}(\g^m)_{\n\rho)} = 0$. Thus, we have verified the nilpotency of the BRST current $j_B$. In an analogous way the antiholomorphicity and nilpotency of ${\widetilde j}_B$ in flat case can be proven.

\subsection{Grading of the target space}\label{grading}

We will consider that the target space for the pure spinor sigma model is a graded manifold \cite{Voronov2001}. The grading is given by the bosonic coordinates $(\l^\a, \lh^\ah)$ and is conventionally called ghost number.

We also require that the generalized metric $\AA$ has grading $(0,0)$ or in other words of ghost number zero
  \begin{equation}
  \AA = (d\ZZ^j\otimes d\ZZ^i)\AA_{ij} + (\frac{\p}{\p\ZZ^j}\otimes d\ZZ^i) \AA_i{}^j + (d\ZZ^j\otimes \frac{\p}{\p\ZZ^i}) \AA^i{}_j + (\frac{\p}{\p\ZZ^j}\otimes\frac{\p}{\p\ZZ^i})\AA^{ij}
  \end{equation}
Thus, all terms proportional to objects such as $d\l^\a\otimes d\l^\b$, $dZ^M\otimes d\l^\a$, $d\l^\a\otimes \frac{\p}{\p Z^M}$,$d\l^\a\otimes \frac{\p}{\p\lh^\bh}$ need to vanish since their coefficient functions are not allowed to have negative ghost number i.e. no inverse powers of $\l^\a$, $\lh^\ah$. Also, remember that all coefficients such as $\AA_{ij}$, $\AA^{ij}$, $\AA_i{}^j$ are functions on the base manifold $M$, in other words, they only depend on variables $(Z^M, \l^\a, \lh^\ah)$ and not on the momentum variables $(P_M, \o_\a,\oh_\ah)$ which may account for negative ghost number. Then the following coefficients (entries in the generalized metric) are zero
  \begin{equation}\label{zeroes}
  \begin{aligned}
  \AA_{M\a} = \AA_{M\ah} = \AA_{\a\b} = \AA_{\a\bh} = \AA_{\ah\bh} = 0 \\
  \AA_\a{}^M = \AA_\a{}^\ah = \AA_\ah{}^M = \AA_\ah{}^\a = 0
  \end{aligned}
  \end{equation}
as well as their transposes. This assertion can be readily verified in the expression for the generalized metric of the Pure Spinor string (\ref{ps-met}).

As a consequence of the ghost number grading, we have (\ref{zeroes}) which together with the fact that for the Pure Spinor case $\AA^\a{}_\b = \d^\a_\b$ and $\AA^\ah{}_\bh = -\d^\ah_\bh$ imply that the vector fields
  \begin{equation}
  \VV = \VV^\a(Z,\l,\lh) \frac{\p}{\p \l^\a}\ ,\quad \widetilde\VV = \widetilde\VV^\ah(Z,\l,\lh) \frac{\p}{\p\lh^\ah}
  \end{equation}
are respectively $(+1)$ and $(-1)$ eigenvectors of $\UU_{(\AA_{PS})}$ for any curved background. This means the eigenspaces $(TM\oplus T^*M)_\pm$ have a non-trivial intersection with the tangent bundle $TM$. Equivalently we could say that the $[\AA_{ij}]$ part of the generalized metric $\AA_{PS}$ is not invertible.

On the other hand, if we were to try to compute the intersections of $(\pm 1)$-eigenspaces of $\UU_{\AA_{PS}}$ with $T^*M$ just as we did for the flat case, we would see that there exist a set of equations that may not have solutions in the generic case, meaning that the intersections $(TM\oplus T^*M)_\pm \cap T^*M$ could lifted to be zero. Then, it is better to study directly the invertibility of the matrix $[\AA^{ij}]$. In the pure spinor case it can be written as $[\AA^{ij}] = [\EE_{\underline i}{}^i]^T [\AA^{\underline i \underline j}] [\EE_{\underline j}{}^j]$ where the underlined labels can be thought as flat indices and $\EE_{\underline i}{}^j$ as vielbein. Thus, $[\AA^{ij}]$ is congruent to $[\AA^{\underline i \underline j}]$
  \begin{equation}\label{congruent}
  [\AA^{ij}] \quad \sim \quad [\AA^{\underline i \underline j}] = \left(
    \begin{array}{c|cccc}
    \eta^{ab} & 0 & 0 & 0 & 0 \\
    \hline
    0 & 0 & P^{\a\bh} & 0 & (\lh \widehat C)^{\bh\a} \\
    0 & - P^{\b\ah} & 0 & (\l C)^{\b\ah} & 0 \\
    0 & 0 & (\l C)^{\a\bh} & 0 & (\l\lh S)^{\a\bh} \\
    0 & (\lh\widehat C)^{\ah \b} & 0 & (\l\lh S)^{\b\ah} & 0
    \end{array}
  \right)
  \end{equation}
and the vielbein $\EE_{\underline i}{}^j$ is
  \begin{equation}
  [\EE_{\underline i}{}^j] = \left(
    \begin{array}{ccc}
    E_A{}^N & E_A{}^N(\l\O_N)^\b & E_A{}^N(\lh\Oh_N)^\bh \\
    0 & \d_\a{}^\b & 0 \\
    0 & 0 & \d_\ah{}^\bh
    \end{array}
  \right)
  \end{equation}
We can now see that the second diagonal block in (\ref{congruent}) is not invertible since its inverse would need to have entries with inverse powers of $\l^\a$ and $\lh^\ah$ which is forbidden by construction on graded manifolds. This means there is not possible to have a second order formulation for the pure spinor sigma model in terms on some tensor $\GG_{ij}$ and $\BB_{ij}$.

\section{Summary and future directions}
In order to obtain a formulation of the Pure Spinor string that treats symmetrically matter and ghost fields, we have studied sigma model actions in their Hamiltonian form. We also showed that by extending the definition of a generalized metric we can treat in a unified manner theories that mix first and second order actions. Then, we verified that the Pure Spinor sigma model action (\ref{ps-action}) can be expressed in Hamiltonian form for any curved background, and the appearing matrix $\AA_{PS}$ satisfies the conditions of a generalized metric. Furthermore, there exist fermionic sections $(\QQ,\ \FF)$ and $(\QQt,\ \FFt)$ that are eigenvectors of $\AA_{PS}$ and generate the BRST currents $j_B$ and $\widetilde j_B$.

Finally, we were able to describe the conditions of holomorphicity and nilpotency of currents in terms of constraints on the associated sections and the generalized metric (\ref{holomorphicity}, \ref{nilpotency}). And, ultimately we used this in the Pure Spinor case to deduce that the background fields $(E_\a{}^M,E_\ah{}^M,\O_{M\a}{}^\b, \Oh_{M\ah}{}^\bh, C^{\a\ah}_\b, {\widehat C}^{\ah\a}_\bh, S^{\a\ah}_{\b\bh})$ satistfy the Type II supergravity constraints (\ref{sugra-nil1}-\ref{sugra-nil3}, \ref{sugra-hol1}-\ref{sugra-hol2}).

A future direction is to take the Pure Spinor sigma model and study the simpler case where the "matter" part of the target space is a complex supermanifold. With the additional supposition that the 2-form $B_{MN}$ is the Kahler form, we are able to consistently associate holomorphic variables in the worldsheet and target space. Thus, we land on a theory that is an honest first order formulation much like in \cite{Losev2005}
  \begin{equation}
  S = \int p_M \pb z^M + \bar p_{\bar M} \p \bar z^{\bar M} + \o_\a\pb\l^\a + \oh_\ah \p\lh^\ah -\frac12 (p_M\ \o_\a) \left(
    \begin{array}{cc}
    A^{M\bar M} & A^{M\ah} \\
    A^{\a \bar M} & A^{\a\ah}
    \end{array}
  \right) \left(
    \begin{array}{c}
    \bar p_{\bar M} \\
    \oh_\ah
    \end{array}
  \right) \nonumber
  \end{equation}
where the matrix $A$ contains all background fields. And, the $BRST$ currents simplify to be generated only by vector fields
  \begin{equation}
    \begin{aligned}
    \QQ_B & = \l^\a E_\a{}^M \frac{\p}{\p z^M} - \l^\a E_\a{}^M \l^\b \O_{M \b}{}^\g \frac{\p}{\p \l^\g} \\
    \QQt_B & = \lh^\ah E_\ah{}^{\bar M} \frac{\p}{\p \zb^{\bar M}} - \lh^\ah E_\ah{}^{\bar M} \lh^\bh \Oh_{{\bar M} \bh}{}^\gh \frac{\p}{\p \lh^\gh}
    \end{aligned}
  \end{equation}
A particular example of this is the Pure Spinor string in a Calabi-Yau background.

\vspace{1cm}

\noindent {\bf Acknowledgements.}
The author would like to thank Andrei Mikhailov for suggesting the problem and useful discussions. This project was financially supported by FAPESP grant 2016/22579-9.

\newpage

\appendix

\section{Block entries of the $\AA_{PS}$ matrix}\label{app}
The matrix $\AA_{PS}$ for the Pure Spinor action depends on the background fields in the following way. The matrix $[\AA^{ij}]$ is
  \begin{equation}
  [\AA^{ij}] = \left(
    \begin{array}{ccc}
    \AA^{MN} & \AA^{M\b} & \AA^{M\bh} \\
    \AA^{\a N} & \AA^{\a\b} & \AA^{\a\bh} \\
    \AA^{\ah N} & \AA^{\ah\b} & \AA^{\ah\bh}
    \end{array}
  \right)
  \end{equation}
where the blocks are
  \begin{align}
  \AA^{MN} & = A^{MN} := (-)^M E_a{}^M E_b{}^N \eta^{ab} + E_\a{}^M E_\ah{}^N P^{\a\ah} - E_\ah{}^M E_\a{}^N P^{\a\ah} \\
  \AA^{M\b} & = -A^{MN}(\l\O_N)^\b - (-)^M (\l C)^{\b\bh} E_\bh{}^M \\
  \AA^{M\bh} & = -A^{MN}(\lh\Oh_N)^\bh - (-)^M (\lh {\widehat C})^{\bh\b} E_\b{}^M \\
  \AA^{\a\b} & = (\l\O_M)^\a A^{MN} (\l\O_N)^\b + (\l C)^{\a\ah} E_\ah{}^M(\l\O_M)^\b + (\l C)^{\b\bh} E_\bh{}^M (\l\O_M)^\a \\
  \AA^{\a\bh} & = (\l\O_M)^\a A^{MN} (\lh\Oh_N)^\bh + (\l C)^{\a\ah} E_\ah{}^M(\lh\Oh_M)^\bh + (\lh{\widehat C})^{\bh\b} E_\b{}^M(\l\O_M)^\a \nonumber \\
  & \quad  -(\l\lh S)^{\a\bh} \\
  \AA^{\ah\bh} & = (\lh\Oh_M)^\a A^{MN} (\lh\Oh_N)^\b + (\lh {\widehat C})^{\ah\a} E_\a{}^M(\lh\Oh_M)^\bh + (\lh {\widehat C})^{\bh\b} E_\b{}^M (\lh\Oh_M)^\ah
  \end{align}
while because of graded-symmetry $\AA^{ij} = (-)^{ij+i+j} \AA^{ji}$ the rest are $\AA^{\a N} = (-)^N \AA^{N \a}$, $\AA^{\ah N} = (-)^N \AA^{N\ah}$ and $\AA^{\ah \b} = \AA^{\b \ah}$. The matrix $[\AA_{ij}]$ take values
  \begin{equation}
  [\AA_{ij}] = \left(
    \begin{array}{ccc}
    \AA_{MN} & 0 & 0 \\
    0 & 0 & 0 \\
    0 & 0 & 0
    \end{array}
  \right)
  \end{equation}
where the non-trivial block $\AA_{MN}$ is
  \begin{equation}
  \AA_{MN} = G_{MN} - B_{ML}A^{LP} B_{PN} - E_{(M|}{}^\ah E_\ah{}^P B_{P|N)} + E_{(M|}{}^\a E_\a{}^P B_{P|N)}
  \end{equation}
And, finally the matrix $[\AA_i{}^j]$ is given by
  \begin{equation}
  [\AA_i{}^j] = \left(
    \begin{array}{ccc}
    \AA_M{}^N & \AA_M{}^\b & \AA_M{}^\bh \\
    0 & \d_\a{}^\b & 0 \\
    0 & 0 & -\d_\ah{}^\bh
    \end{array}
  \right)
  \end{equation}
where the non-trivial blocks are
  \begin{align}
  \AA_M{}^N & = (-B_{MP} A^{PN} - E_M{}^\ah E_\ah{}^N + E_M{}^\a E_\a{}^N) \\
  \AA_M{}^\b & = (\d_M{}^N - \AA_M{}^N)(\l\O_N)^\b - (\l C)^{\b\bh} E_\bh{}^N B_{NM} \\
  \AA_M{}^\bh & = - (\d_M{}^N + \AA_M{}^N)(\lh\Oh_N)^\bh - (\lh{\widehat C})^{\bh\b} E_\b{}^N B_{NM}
  \end{align}
The other off-diagonal matrix block $[\AA^i{}_j]$ is just the graded-transpose of $[\AA_j{}^i]$ i.e. $\AA^i{}_j = (-)^{i+ij} \AA_j{}^i$.

\section{Conformal invariance and the orthosymplectic supergroup $OSp(d,d|2s)$}\label{app-ortho}

\subsection{$SO(1,1)$ worldsheet symmetry}
The action (\ref{action}) does not show an explicit $SO(1,1)$ invariance. However, to have this symmetry it is enough that the matrix $\AA$ satisfies the $O(d,d)$ condition for the super-case i.e.
  \begin{equation}\label{odd}
  \left(
    \begin{array}{cc}
    \AA_{ik} & \AA_i{}^k \\
    \AA^i{}_k & \AA^{ik}
    \end{array}
  \right) \left(
    \begin{array}{cc}
    0 & \d^k{}_l \\
    \d_k{}^l & 0
    \end{array}
  \right) \left(
    \begin{array}{cc}
    \AA_{lj} & \AA_l{}^j \\
    \AA^l{}_j & \AA^{lj}
    \end{array}
  \right) = \left(
    \begin{array}{cc}
    0 & \d_i{}^j \\
    \d^i{}_j & 0
    \end{array}
  \right)
  \end{equation}
To see this, we will transform the action variables under the $SO(1,1)$ group. On the worldsheet variables the infinitesimal transformation is
  \begin{equation}
  \d \left(
    \begin{array}{c}
    t \\
    \s
    \end{array}
  \right) = \ve\left(
    \begin{array}{cc}
    0 & 1 \\
    1 & 0
    \end{array}
  \right) \left(
    \begin{array}{c}
    t \\
    \s
    \end{array}
  \right)
  \end{equation}
while this implies that
  \begin{equation}
    \begin{aligned}
    \d(\p_t \ZZ^i) & = - \ve \p_\s \ZZ^i \\
    \d(\p_\s \ZZ^i) & = -\ve\p_t \ZZ^i \\
    \d \PP_i & = \ve \PP^\s_i
    \end{aligned}
  \end{equation}
where $\PP^\s_i = \frac{\d L}{\d(\p_\s\ZZ^i)}$. From the action (\ref{action}), we can compute expressions for $\p_t \ZZ^i$ and $\PP^\s_i$
  \begin{equation}\label{transform}
    \begin{aligned}
    \p_t \ZZ^i & = \p_\s \ZZ^j \AA_j{}^i + \PP_j \AA^{ji} \\
    \PP^\s_i & = - \p_\s \ZZ^j \AA_{ji} - \PP_j \AA^j{}_i
    \end{aligned}
  \ \Rightarrow\ (\p_t\ZZ^i\ \ (-\PP^\s_i)) = (\p_\s\ZZ^j\ \ \PP_j)\left(
    \begin{array}{cc}
    \AA_{jk} & \AA_j{}^k \\
    \AA^j{}_k & \AA^{jk}
    \end{array}
  \right) \left(
    \begin{array}{cc}
    0 & \d^k{}_i \\
    \d_k{}^i & 0 
    \end{array}
  \right)
  \end{equation}
This last expression will be useful since several terms in the variation of the action $S = \int (\p_t \ZZ^i \PP_i - H)$ can be written in that form. 

The kinetic term in the action $\p_t \ZZ^i \PP_i$ varies as
  \begin{equation}\label{transf-1}
  \d(\p_t \ZZ^i \PP_i) = -\ve \p_\s\ZZ^i \PP_i + \ve\p_t\ZZ^i \PP^\s_i 
  \end{equation}
whereas this last term can be rewritten as
  \begin{align}
  \p_t\ZZ^i \PP^\s_i & = -\frac12 (\p_t\ZZ\ (-\PP^\s))\left(
    \begin{array}{cc}
    0 & 1 \\
    1 & 0
    \end{array}
  \right) \left(
    \begin{array}{c}
    \p_t\ZZ \\
    -\PP
    \end{array}
  \right) \nonumber \\
  & = -\frac12 (\p_\s \ZZ\ \PP) \AA \left(
    \begin{array}{cc}
    0 & 1 \\
    1 & 0
    \end{array}
  \right) \left(
    \begin{array}{cc}
    0 & 1 \\
    1 & 0
    \end{array}
  \right) \left(
    \begin{array}{cc}
    0 & 1 \\
    1 & 0
    \end{array}
  \right) \AA \left(
    \begin{array}{c}
    \p_\s \ZZ \\
    \PP
    \end{array}
  \right) \nonumber \\
  & = -\frac12 (\p_\s \ZZ\ \PP)\left(
    \begin{array}{cc}
    0 & 1 \\
    1 & 0
    \end{array}
  \right) \left(
    \begin{array}{c}
    \p_\s \ZZ \\
    \PP
    \end{array}
  \right) = - \p_\s \ZZ^i \PP_i \label{transf-2}
  \end{align}
Here we used (\ref{transform}) and (\ref{odd}) in the second and third equalities respectively. Similarly, we can compute the how the Hamiltonian part of the action varies
  \begin{align}
  \d H & = -\ve\ \left(\p_t\ZZ\ \ (-\PP^\s)\right)\AA\left(
    \begin{array}{c}
    \p_\s\ZZ \\
    \PP
    \end{array}
  \right) \nonumber \\
  & = -\ve \left(\p_\s \ZZ\ \PP\right) \AA \left(
    \begin{array}{cc}
    0 & 1 \\
    1 & 0
    \end{array}
  \right) \AA \left(
    \begin{array}{c}
    \p_\s \ZZ \\
    \PP
    \end{array}
  \right) \nonumber \\
  & = -\ve\ \left(\p_\s \ZZ\ \PP\right)\left(
    \begin{array}{cc}
    0 & 1 \\
    1 & 0
    \end{array}
  \right)\left(
    \begin{array}{c}
    \p_\s \ZZ \\
    \PP
    \end{array}
  \right) \label{hamiltonian-transf}
  \end{align}
where in the second equality we used (\ref{transform}); and in the third one, the $O(d,d)$ property of $\AA$ (\ref{odd}). Finally, from (\ref{transf-1}), (\ref{transf-2}) and (\ref{hamiltonian-transf}) the variation of the action $ S = \int(\p_t \ZZ^i \PP_i - H)$ reduces to
  \begin{equation}
  \d S = \ve \int (- 2\p_\s \ZZ^i \PP_i) + \left(\p_\s \ZZ\ \PP\right)\left(
    \begin{array}{cc}
    0 & 1 \\
    1 & 0
    \end{array}
  \right)\left(
    \begin{array}{c}
    \p_\s \ZZ \\
    \PP
    \end{array}
  \right) = 0
  \end{equation}
This ends up the proof the invariance of the action under $SO(1,1)$.

\subsection{Manifest conformal invariance}
We take the action (\ref{action}) and show that the $O(d,d)$ condition implies that the action can be written in a form that shows manifest conformal invariance. We change the indices in this subsection and now denote the letters $M,N,P,Q,...$ as curved labels and $A,B,C,D,...$ as flat ones.

Consider that the tensor $\AA^{MN}$ can be written as
  \begin{equation}
  \AA^{MN} = E_A{}^M \eta^{AB} E_B{}^N\ ,\qquad [\eta^{AB}] = \left(
    \begin{array}{cc}
    \eta^{ab} & 0 \\
    0 & \eta^{\a\b}
    \end{array}
  \right)
  \end{equation}
where $\eta^{ab}$ is invertible and $\eta^{\a\b}$ is not. If we were in the purely bosonic case, we could make $\eta^{\a\b} = 0$ and the dimensions of the matrix $\eta^{ab}$ would be the rank of $\AA^{MN}$. Let's define the following fields
  \begin{equation}
    \begin{aligned}
    d_A & := E_A{}^M \PP_M \\
    \Pi_t^A & := \p_t \ZZ^M E_M{}^A \\
    \Pi_\s^A & := \p_\s \ZZ^M E_M{}^A \\
    \left( 
      \begin{array}{cc}
      A_{AB} & A_A{}^B \\
      A^A{}_B & A^{AB}
      \end{array}
    \right) & :=  \left(
      \begin{array}{cc}
      E_A{}^M & 0 \\
      0 & E_M{}^A
      \end{array}
    \right) \left(
      \begin{array}{cc}
      \AA_{MN} & \AA_M{}^N \\
      \AA^M{}_N & \AA^{MN}
      \end{array}
    \right) \left(
      \begin{array}{cc}
      E_B{}^N & 0 \\
      0 & E_N{}^B
      \end{array}
    \right)
    \end{aligned}
  \end{equation}
This new matrix with flat indices also satisfies the $O(d,d)$ condition. The action (\ref{action}) can be written as
  \begin{equation}
  S = \int \Pi_t^A d_A - \frac12 (\Pi_\s^A\ d_A) \left(
    \begin{array}{cc}
    A_{AB} & A_A{}^B \\
    A^A{}_B & A^{AB}
    \end{array}
  \right) \left(
    \begin{array}{c}
    \Pi_\s^B \\
    d_B
    \end{array}
  \right)
  \end{equation}
while the $O(d,d)$ condition can be broken down into
  \begin{align}
  A^{AC} A_C{}^B + A^A{}_C A^{CB} & = 0 \quad \left\{
    \begin{aligned}
    A_a{}^c\eta_{cb} + \eta_{ac} A^c{}_b = 0 \\
    A_a{}^\b + \eta_{ab} A^b{}_\a \eta^{\a\b} = 0 \\
    \eta^{\a\g} A_\g{}^\b + A^\a{}_\g \eta^{\g\b} = 0
    \end{aligned} 
  \right. \\
  A_A{}^C A_C{}^B + A_{AC} A^{CB} & = \d_A{}^B \quad \left\{ 
    \begin{aligned}
    A_{ab} - A_a{}^c\eta_{cd} A_b{}^d + A_a{}^\a A_\a{}^c \eta_{cb} = \eta_{ab} \\
    A_\a{}^\g A_\g{}^\b + (A_{\a\g} - A_\a{}^a\eta_{ab}A^b{}_\g)\eta^{\g\b} = \d_\a{}^\b \\
    A_a{}^b A_b{}^\b + (A_{a\g} + \eta_{ab}A^b{}_\a A^\a{}_\g)\eta^{\g\b} = 0 \\
    A_\a{}^\b A_\b{}^b \eta_{ba} + A_\a{}^c A_c{}^b \eta_{ba} + A_{\a a} = 0
    \end{aligned}
  \right. \\
  A_A{}^C A_{CB} + A_{AC} A^C{}_B & = 0
  \end{align}

Using the equations of motion for $\PP_M$, we obtain the following expressions
  \begin{equation}
  0 = \frac{\d S}{\d \PP_M} E_M{}^A = \Pi_t^A - \eta^{AB} d_B - \Pi_\s^B A_B{}^A
  \end{equation}
that can be used to solve only for $d_a$
  \begin{align}
  d_a & = \Pi^b_t\eta_{ba} - \Pi_\s^A A_A{}^b \eta_{ba} \\
   d_\b\eta^{\b\a} & = \Pi_t^\a - \Pi^A_\s A_A{}^\a
  \end{align}
We use the first equation to replace it back into the action and obtain
  \begin{align}
  S & = \int \left[ \frac12 \left( \Pi_t^a - \Pi_\s^A A_A{}^a\right)\eta_{ab}\left( \Pi_t^b - \Pi_\s^A A_A{}^b\right) - \frac12 \Pi^A_\s A_{AB} \Pi^B_\s \right. \nonumber \\
  & \qquad  \left. + \left(\Pi_t^\a - \Pi_\s^A A_A{}^\a\right) d_\a - \frac12 d_\a \eta^{\a\b} d_\b \right]\\
  & = \int \left[ \frac12 \Pi_t^a\eta_{ab}\Pi_t^b -\frac12 \Pi_\s^a(A_{ab}-A_a{}^c\eta_{cd}A_d{}^b)\Pi_\s^b - \frac12 \Pi_t^a \eta_{ab} A^b{}_c\Pi^c_\s - \frac12 \Pi^a_\s A_a{}^c\eta_{cb}\Pi_t^b \right. \nonumber \\
  & \qquad \left. -\frac12 \Pi_\s^\a(A_{\a\b} - A_\a{}^a\eta_{ab}A_\b{}^b)\Pi_\s^\b + \frac12 d_\a \eta^{\a\b} d_\b - \Pi_\s^a A_{a\a} \Pi_\s^\a - \Pi_\s^\a A_\a{}^a\eta_{ab}(\Pi_t^b - \Pi^c_\s A_c{}^b) \right]
  \end{align}

Finally, we define $d_A^\s := \frac{\d L}{\d \Pi_\s^A}$ and using the value for
  \begin{equation}
  d_\a^\s = - (A_{\a\b} - A_\a{}^a\eta_{ab}A_\b{}^b) \Pi^\b_\s - A_\a{}^\b d_\b + A_\a{}^\b A_\b{}^b \eta_{ba} \Pi^a_\s - A_\a{}^a \eta_{ab}\Pi^b_t
  \end{equation}
we can express the action as
  \begin{align}
  S & = \int \frac12 \Pi_t^a\eta_{ab} \Pi_t^b - \frac12 \Pi_\s^a\eta_{ab}\Pi_\s^b - \frac12 \Pi_t^a\eta_{ac}A^c{}_b\Pi^b_\s - \frac12 \Pi_\s^a A_a{}^c\eta_{cb}\Pi^b_t \nonumber \\
  & \qquad + \frac12 \Pi^\a_\s d^\s_\a + \frac12 \Pi_t^\a d_\a^t + \frac12 \Pi^\a_t A_\a{}^b\eta_{ba}\Pi_\s^a - \frac12 \Pi_\s^\a A_\a{}^b\eta_{ba}\Pi_t^a \\
  & = \int 2 \Pi_z^a (\eta_{ab} + A_a{}^c\eta_{cb})\Pi_\zb^b + \Pi^\a_z d_{\zb \a} + \Pi^\a_\zb d_{z \a} + \Pi^\a_z A_\a{}^b \eta_{ba} \Pi^a_\zb - \Pi^\a_\zb A_\a{}^b\eta_{ba}\Pi^a_z
  \end{align}
We have arrived at an action that shows manifest conformal invariance on each term. The $(z,\zb)$ worldsheet quantities are defined as
  \begin{equation}
  \begin{aligned}
  \Pi_t^A & = \Pi^A_\zb + \Pi^A_z \\
  \Pi_\s^A & = \Pi^A_\zb - \Pi^A_z \\
  d_A^t & = d_z + d_\zb \\
  d_A^\s & = d_z - d_\zb
  \end{aligned}
  \end{equation}

\subsection{Orthosymplectic supergroup $OSp(d,d|2s)$}
The $O(d,d)$ condition (\ref{odd-condition}) that a generalized metric $A$ satisfies was naturally extended to the super-case (\ref{osp-condition}), and we have just proved that it is a sufficient condition to have an $SO(1,1)$ worldsheet symmetry of the action. However, we can give this condition another interpretation if we write the matrix representation as a super-matrix i.e. a matrix with bosonic entries in the block diagonals and fermionic ones in the off-diagonal. We achieve this by changing the order of the entries of the sections contracting with $\AA$ in the action (\ref{action}) i.e.
  \begin{equation}
  \left(
    \begin{array}{c}
    \p_\s \ZZ^i \\
    \PP_i
    \end{array}
  \right) = \left(
    \begin{array}{c}
    \p_\s x^m \\
    \p_\s \t^\m \\
    P_m \\
    P_\m
    \end{array}
  \right)\quad \longrightarrow \quad \left(
    \begin{array}{c}
    \p_\s x^m \\
    P_m \\
    \p_\s \t^\m \\
    P_\m
    \end{array}
  \right)
  \end{equation}
where $(x^m)$ and $(\t^\m)$ simply represent bosonic and fermionic coordinates parametrizing the target space supermanifold $M$. Then, the action (\ref{action}) takes the form
  \begin{equation}
  S = \int (\p_t x^m P_m + \p_t \t^\m P_\m) -\frac12 (\p_\s x^m\ P_m\ \p_\s\t^\m\ P_\m)\AA' \left(
    \begin{array}{c}
    \p_\s x^n \\
    P_m \\
    \p_\s \t^\m \\
    P_\m
    \end{array}
  \right)
  \end{equation}
where
  \begin{equation}
  \AA' = \left(
    \begin{array}{c|cc|c}
    \d_m{}^n & & & \\
    \hline
    & & \d^m{}_n & \\
    & \d_\m{}^\n & & \\
    \hline
    & & & \d^\m{}_\n
    \end{array}
  \right) \AA \left(
    \begin{array}{c|cc|c}
    \d^m{}_n & & & \\
    \hline
    & & -\d^\m{}_\n & \\
    & \d_m{}^n & & \\
    \hline
    & & & \d_\m{}^\n
    \end{array}
  \right)
  \end{equation}
Then the condition (\ref{osp-condition}) on $\AA$ implies that $\AA'$ will satisfy
  \begin{equation}
  (\AA')^{st} \left(
    \begin{array}{cc|cc}
    0 & 1 & & \\
    1 & 0 & & \\
    \hline
    & & 0 & -1 \\
    & & 1 & 0
    \end{array}
  \right) \AA' = \left(
    \begin{array}{cc|cc}
    0 & 1 & & \\
    1 & 0 & & \\
    \hline
    & & 0 & -1 \\
    & & 1 & 0
    \end{array}
  \right)
  \end{equation}
where $(\ldots)^{st}$ is the super-transpose acting on a supermatrix. This is no more than the defining property of the orthosymplectic supergroup $OSp(d,d|2s)$; then, the supermatrix $\AA'$ belongs to $OSp(d,d|2s)$.

\section{Generalized metrics}\label{gen-met}
In this appendix we extend the definition of a generalized metric to include cases with indefinite signature. We will show that this generalized metric with indefinite signature also has other equivalent definitions in terms of subbundles of $TM\oplus T^*M$. Afterwards, we study what is usually known in the literature as generalized metric on $TM\oplus T^*M$.

\subsection{Generalized metric with signature $(d+a,d-a)$}
Given a manifold $M$ of dimension $d$, we can construct its {\it generalized} tangent bundle $TM\oplus T^*M$ where its sections are pairs of a vector and $1$-form fields. In this bundle there exists a canonical inner product of signature $(d,d)$
  \begin{equation}
    \begin{aligned}
    \langle \cdot, \cdot\rangle : \G(TM\oplus T^*M) \times \G(TM\oplus T^*M)\quad &\rightarrow \quad C^\infty (M) \\
    (V_1\ F_1) \times (V_2\ F_2)\hspace{1.5cm} & \mapsto\quad \i_{V_1} F_2 + \i_{V_2} F_1
    \end{aligned}
  \end{equation}
Once we have chosen a coordinate system $(x^m)$ for $M$, we have a basis on $TM\oplus T^*M$ and a matrix representation for $\langle\cdot,\cdot \rangle$
  \begin{equation}
  \langle\cdot,\cdot \rangle \quad \longrightarrow \quad \left(
    \begin{array}{cc}
    0 & \d_m{}^n \\
    \d^m{}_n & 0
    \end{array}
  \right)
  \end{equation}
Now we define a generalized metric with a signature as
  \begin{definition}
  A generelized metric $A$ with signature is defined as a fiber-wise bilinear form on $TM\oplus T^*M$
    \begin{equation}
    A: \G(TM\oplus T^*M) \times \G(TM\oplus T^*M) \longrightarrow C^\infty(M)
    \end{equation}
  such that it is symmetric, has signature $(d+a, d-a)$ and satisfies the $O(d,d)$ condition 
    \begin{equation}
    A^t \left(
      \begin{array}{cc}
      0 & 1 \\
      1 & 0
      \end{array}
    \right) A = \left(
      \begin{array}{cc}
      0 & 1 \\
      1 & 0
      \end{array}
    \right)
    \end{equation}
  \end{definition}

  \begin{definition}
  A generalized metric on $E = TM\oplus T^*M$ is an automorphism $U:E\rightarrow E$ such that
    \begin{itemize}
    \item $U^2 = 1$
    \item $A_U(-,-) := \langle U(-) ,- \rangle$ defines a fiber-wise inner product of signature $(d+a, d-a)$
    \end{itemize}
  \end{definition}

We can see that a generalized metric $U$ will have the following properties. First of all, it will be diagonalizable with the only eigenvalues being $+1$ and $-1$. Also, $U$ will be orthogonal w.r.t the inner product $\langle\cdot,\cdot\rangle$ because of the symmetry property of $A_U$. And finally, the eigenspaces of $U$ (denoted by $E_\pm$) will be orthogonal to each other w.r.t to the inner product $\langle\cdot,\cdot\rangle$ i.e. $E_+ \perp_{\langle,\rangle} E_-$.

  \begin{definition}
  A generalized metric with signature on $E=TM\oplus T^*M$ is a subbundle $E_+$ such that $\langle\cdot,\cdot\rangle|_{E_+}$ has signature $(p,q)$.
  \end{definition}
Notice that the orthogonal complement to $E_+$ will satisfy $E = E_+ \oplus (E_+)^\perp$ since $\langle\cdot,\cdot\rangle$ is non-degenerate on $E_+$. Furthermore, $\langle\cdot,\cdot\rangle$ will also be non-degenerate on $(E_+)^\perp$ since $E_+\cap (E_+)^\perp = \{0\}$, and its signature there will be $(d-p,d-q)$. This can be seen by taking a normalizing basis for each $E_+$ and $(E_+)^\perp$ and realizing that the signature of $\langle\cdot,\cdot\rangle$ on the full space is $(d,d)$.

  \begin{theorem}
  All three definitions are equivalent.
  \end{theorem}
  \begin{proof}
  Showing that Def.1 and Def.2 are equivalent is almost trivial since $U$ and $A$ are related by the relation $A(-,-)= \langle U(-),-\rangle$ or in matrix form
    \begin{equation}
    A = U\left(
      \begin{array}{cc}
      0 & 1 \\
      1 & 0
      \end{array}
    \right)
    \end{equation}
  which connects $U$ and $A$ in a one-to-one relation.
  
  To show that Def.2 implies Def.3, we take the $(+1)$-eigenspace of $U$ as our subbundle $E_+$. And, since $E = E_+\oplus E_-$ with $E_+ \perp_{\langle,\rangle} E_-$, we have that $\langle\cdot,\cdot\rangle$ is non-degenerate on $E_+$ and $E_-$, this means it has respectively signatures $(p,q)$ and $(d-p,d-q)$. We can take a basis of each $E_+$ and $E_-$, $\{e^+_i\}$ and $\{e^-_m\}$ respectively, such that
  \begin{equation}
  \langle e^+_i,e^+_j\rangle = (\eta^+)_{ij}\ ,\quad \langle e^-_m, e^-_n\rangle = (\eta^-)_{mn}
  \end{equation}
where $\eta^+$ and $\eta^-$ are the signature matrices $(p,q)$ and $(n-p,n-q)$. Then we can compute the matrix representation of $A_U$ on the basis $\{e^+_i,e^-_m\}$
  \begin{equation}\label{matrixrep}
  A_U = \left(
    \begin{array}{cc}
    (\eta^+)_{ij} & 0 \\
    0 & -(\eta^-)_{mn}
    \end{array}
  \right)
  \end{equation}
which means that $A_U$ has signature $(d + (p-q), d-(p-q))$ i.e. $a = (p-q)$. Finally, the signature of $\langle\cdot,\cdot\rangle$ on subbundle $E_+$ is given by $p = (\dim E_+ + a)/2$ and $q = (\dim E_+ - a)/2$.

  Similarly, to show that Def.3 implies Def.2 we proceed as follows. Take the orthogonal complement to $E_+$ denoted by $E_- := (E_+)^\perp$. Since $\langle\cdot,\cdot\rangle$ is non-degenerate on $E_+$ with signature $(p,q)$ we have that $E= E_+\oplus E_-$ and it is also non-degenerate on $E_-$ with signature $(d-p,d-q)$. Then we can define the operator $U:E\rightarrow E$ as $U(v^\pm) = \pm v^\pm$ for $v^\pm \in E_\pm$ (and by linear extension to the rest of the space). We can now prove that this $U$ has the properties of Def.2. Thus, obviously $U^2 =1$. And, the signature of $A_U$ can be determined by taking again the basis of $E_+$ $\{e^+_i\}$ and of $E_-$ $\{e^-_m\}$ and computing the matrix representation of $A_U$ on $\{e^+_i, e^-_m\}$. We obtain the same as in (\ref{matrixrep}) which means that the signature of $A_U$ is of type $(d+a,d-a)$ with $a = p-q$.
  \end{proof}

\subsection{(Positive definite) Generalized metric}
The particular case where $a=d$ in Definition 1 and 2 or equivalently $(p,q)=(d,0)$ in Definition 3 describes what is simply known in the literature as {\it generalized metric}. The implication of this is that the $(\pm 1)$-eigenspaces of $U_A$ can be solved as a the graphs of $(\pm G + B)$ where $G$ is a Riemannian metric and $B$ a two-form on $M$.

  \begin{theorem}
  Given a generalized metric A on $TM\oplus T^*M$. The eigenspaces of $U_A$ are $(TM\oplus T^*M)_\pm = graph(\pm G+ B)$.
  \end{theorem}
  \begin{proof}
  The canonical inner product is positive definite on $(TM\oplus T^*M)_+$ and negative definite on $(TM\oplus T^*M)_-$ as we can see in
    \begin{equation}
    A(s_\pm, s_\pm) = \langle U_A(s_\pm), s_\pm \rangle = \pm \langle s_\pm, s_\pm\rangle > 0
    \end{equation}
  for $s_\pm \in (TM\oplus T^*M)_\pm$. Then, these eigenspaces are forced to be of dimension $d=dim(M)$, since they are complementary and the signature of $\langle\cdot,\cdot\rangle$ on $TM\oplus T^*M$ is $(d,d)$. Also, the intersection of $(TM\oplus T^*M)_\pm$ with either subbundles $TM$ or $T^*M$ is trivial since they are isotropic (i.e. for every section $s$ belonging to them we have $\langle s, s\rangle = 0$). Finally, these last two facts imply that the bundles $(TM\oplus T^*M)_\pm$ can be expressed as the graphs of endomorphisms $\mathcal A_\pm$ in $M$ i.e. $(TM\oplus T^*M)_\pm = graph(\mathcal A_\pm):= \{(V,F)| F=\i_V \mathcal A_\pm\}$. Furthermore, because the $(\pm 1)$- eigenspaces are orthogonal w.r.t. the inner product $\langle\cdot,\cdot\rangle$ we have
    \begin{equation}
    \mathcal A_- = -(\mathcal A_+)^t
    \end{equation}
  Thus, if we define $G$ as the symmetric part of $\mathcal A_+$ and $B$ as the antisymmetric one, we can write $(TM\oplus T^*M)_\pm = graph(\pm G+ B)$.
  \end{proof}

By knowing that $(TM\oplus T^*M) = (TM\oplus T^*M)_+ \oplus (TM\oplus T^*M)_- = graph(+G+B)\oplus graph(-G+B)$, we can easily decompose any section onto those subbundles 
  \begin{equation}
  (V,F)=(V,F)_+ + (V,F)_-\quad\text{where}\quad
    \begin{aligned}
    (V,F)_+ = (V_1, \i_{V_1} (+G+B)) \\
    (V,F)_- = (V_2, \i_{V_2} (-G+B))
    \end{aligned}
  \end{equation}
then, solving for $V_1$ and $V_2$ we obtain
  \begin{equation}
    \begin{aligned}
    V_1 = \frac12 ( V + F G^{-1} - \i_V B G^{-1} ) \\
    V_2 = \frac12 ( V - F G^{-1}  + \i_V B G^{-1} )
    \end{aligned}
  \end{equation}
which implies the very well-known form of a generalized metric
  \begin{equation}
  A = \left(
    \begin{array}{cc}
    G - BG^{-1} B & - BG^{-1} \\
    G^{-1} B & G^{-1}
    \end{array}
  \right)
  \end{equation}

\section{Super-geometry conventions}
In this appendix we give our conventions to deal with the space of super differential forms. We consider objects in this space to be $\ZZ\times\ZZ_2$-graded where the $\ZZ$-grading is given by usual one on differential forms and the $\ZZ_2$-grading corresponds to its bosonic or fermionic nature. For example, given two super differential forms $A$ and $B$ graded $(p,|A|)$ and $(q,|B|)$ respectively, the exterior product satisfies
  \begin{equation}
  A\wedge B = (-)^{pq + |A||B|} B\wedge A
  \end{equation}
In this convention the $\ZZ\times\ZZ_2$ grading of the exterior derivative%
  \footnote{
  We only consider the left exterior derivative i.e. $d(A_p\wedge B_q) = dA_p\wedge B_q + (-)^{p}A_p\wedge dB_q$.
  }
$d$ is $(1,0)$ and the one for the interior product $\i_V$ is $(-1,|V|)$ where $|V|$ is the $\ZZ_2$-grading of the vector field $V$. Furthermore, Cartan's formula is
  \begin{equation}
  \Lc_V = [\i_V, d] = \i_V d - (-)^{(-1)1+|V|0} d \i_V = \i_V d + d \i_V
  \end{equation}
Given a coordinate system $\{Z^M \}$ in the supermanifold, we define the components of a $p$-form $B$ to be $B_{M_1\ldots M_p}:= B(\p_{M_1},\ldots,\p_{M_p})$. Thus, we can write $B$ in the basis $\{dZ^M\}$ as
  \begin{equation}
  B = dZ^{M_p}\wedge\ldots\wedge dZ^{M_1} \frac{1}{p!} B_{M_1\ldots M_p}
  \end{equation}

In section \ref{sect-sugra} we use the holomorphicity equations to obtain the type II SUGRA constraints. This computation required the following expressions for the action of the Lie derivative on tensors of the form $A_{MN}$, $A^{MN}$ and $A_M{}^N$
  \begin{align}
  (\Lc_V A)_{MN} & = V^P \p_P A_{MN} + (-)^{VM} \p_M V^P A_{PN} + (-)^{MN + VN} \p_N V^P A_{PM} \\
  (\Lc_V A)^{MN} & = V^P \p_P A^{MN} - (-)^{VM} A^{MP} \p_P V^N - (-)^{M+N+MN + VN} A^{NP} \p_P V^M \\
  (\Lc_V A)_M{}^N & = V^P \p_P A_M{}^N - (-)^{VM} A_M{}^P \p_P V^N + (-)^{VM} \p_M V^P A_P{}^N
  \end{align}
The SUGRA constraints are expressed in terms of the torsion and curvature. These tensors depend on local frames (vielbein) $E^A$ and a connection $\O_A{}^B$
  \begin{align}
  E^A & = dZ^M (E^A)_M (-)^{MA} =: dZ^M E_M{}^A \\
  \O_A{}^B & = dZ^M (\O_A{}^B)_M (-)^{M(A+B)} =: dZ^M \O_{MA}{}^B
  \end{align}
Then we define the torsion and curvature $2$-forms as covariant exterior derivatives
  \begin{align}
  T^A & = D E^A = dE^A - E^B \wedge \O_B{}^A \\
  R_A{}^B & = D \O_A{}^B = d\O_A{}^B  - \O_A{}^C \wedge \O_C{}^B
  \end{align}
We consider the labels $M,N,P,Q,...$ to be curved indices while $A,B,C,D,...$ to be flat indices. The torsion and curvatures are expressed in components as
  \begin{align}
  T_{AB}{}^C & = (E_A{}^M \p_M E_B{}^N E_N{}^C - E_A{}^M \O_{MB}{}^C) - (-)^{AB}(A\leftrightarrow B) \\
  R_{MNC}{}^D & = (\p_M \O_{NC}{}^D - (-)^{N(C+E)}\O_{MC}{}^E \O_{NE}{}^D) - (-)^{MN} (M\leftrightarrow N)
  \end{align}
For the pure spinor case, we have
  \begin{equation}
  \O_{MA}{}^B = \left(
    \begin{array}{ccc}
    0 & 0 & 0 \\
    0 & \O_{M\a}{}^\b & 0 \\
    0 & 0 & \Oh_{M\ah}{}^\bh
    \end{array}
  \right)
  \end{equation}
Then the expressions for the curvature are simplified to
  \begin{align}
  R_{AB\a}{}^\b & = (-)^{NB} E_A{}^N E_B{}^M \left[ (\p_M \O_{N\a}{}^\b - \O_{M\a}{}^\g \O_{N\g}{}^\b) - (-)^{MN} (M\leftrightarrow N)\right] \\
  \widehat R_{AB\ah}{}^\bh & = (-)^{NB} E_A{}^N E_B{}^M \left[ (\p_M \Oh_{N\ah}{}^\bh - \Oh_{M\ah}{}^\gh \Oh_{N\gh}{}^\bh) - (-)^{MN} (M\leftrightarrow N)\right]
  \end{align}

\end{document}